%% file: main.tex
\documentclass{article}

\usepackage[final, nonatbib]{neurips_2023}
\usepackage{balance}
\usepackage[%
    style=numeric-comp,
    sorting=none,
    natbib=true,  
    backend=biber,
    sortcites=true,
    doi=false,
    url=false,
    giveninits=true,
    mincitenames=1,
    maxcitenames=1, 
    minbibnames=3, 
    maxbibnames=4, 
    hyperref
]{biblatex}
\bibliography{bibtex/external, bibtex/Swarmlab}

\usepackage[utf8]{inputenc} 
\usepackage[T1]{fontenc}    
\usepackage{hyperref}       
\usepackage{url}            
\usepackage{booktabs}       
\usepackage{amsfonts}       
\usepackage{nicefrac}       
\usepackage{microtype}      
\usepackage{xcolor}         

\usepackage[pdftex]{graphicx}
\usepackage{amsmath}
\usepackage{amsfonts}
\usepackage{amsthm}
\usepackage{xfrac}
\usepackage{bm}
\usepackage{subcaption}
\usepackage{wrapfig}

\usepackage{algorithm}
\usepackage{algpseudocode}

\newtheorem{theorem}{Theorem}[section]
\newtheorem{lemma}[theorem]{Lemma}
\newtheorem*{lemma-non}{Lemma}

\newcommand{\figref}[1]{Fig.~\ref{#1}}  
\newcommand{\tabref}[1]{Table~\ref{#1}}

\newcommand{\secref}[1]{Sec.~\ref{#1}}

\newcommand{\appref}[1]{Appendix~\ref{#1}}
\newcommand{\algoref}[1]{Alg.~\ref{#1}}
\newcommand{\lmaref}[1]{Lemma~\ref{#1}}
\newcommand{\Eqref}[1]{(\ref{#1})}

\usepackage{makecell}
\usepackage{tabularray} 
\UseTblrLibrary{booktabs}

\PassOptionsToPackage{maxcitenames=3}{bibtex}

\algrenewcommand\alglinenumber[1]{\tiny #1:}

\title{Task-aware Distributed Source Coding \\ under Dynamic Bandwidth}

\author{
    Po-han Li$^{1}$\thanks{Equal contribution; order decided randomly. Correspondence to \{pohanli, sravan.ankireddy\}@utexas.edu.}
    \quad
    Sravan Kumar Ankireddy$^{1*}$ 
    \quad
    Ruihan Zhao$^{1}$
    \quad
    Hossein Nourkhiz Mahjoub$^{2}$
    \AND
    Ehsan Moradi-Pari$^{2}$
    \quad
    Ufuk Topcu$^{1}$
    \quad
    Sandeep Chinchali$^{1}$
    \quad
    Hyeji Kim$^{1}$
    \AND \vspace{-2em}
    \\
    $^{1}$The University of Texas at Austin 
    $^{2}$Honda Research Institute USA 
}

\begin{document}

\maketitle
\input{section/preamble}

\begin{abstract}
\input{section/abstract}
\end{abstract}

\section{Introduction}
\input{section/intro_v0}

\section{Problem Formulation}
\input{section/problem_formulation}

\section{Theoretical Analysis}
\input{section/DPCA}

\section{Neural Distributed Principal Component Analysis}
\input{section/distributed_autoencoder}

\section{Experiments}
\input{section/experimental_results}

\section{Related Work}
\input{section/related_work}

\section{Conclusion and Future Work}
\input{section/conclusion}

\section*{Acknowledgement}
\input{section/acknowledgement}


\clearpage
\balance
\printbibliography


\newpage
\appendix
\begin{center}
{\bf {\Large Appendix} }
\end{center}
\appendix    
\begin{description}
    \item[\appref{sec:aproof}] provides proofs of lemmas.
    \item[\appref{sec:CEO}] provides additional details on the CEO problem for Gaussian sources.
    \item[\appref{sec:4_sources}] provides an additional experiment of NDPCA with more sources.
    \item[\appref{sec:app_datasets}] provides details of the datasets.
    \item[\appref{sec:app_imple}] provides implementation details.
    \item[\appref{sec:ablation}] provides an ablation study with norms on training loss, the DPCA module, and the views in locate and lift.
\end{description}

\input{section/proof}
\input{section/details}

\end{document}

%% file: section/preamble.tex
\newcommand{\philip}[1] {{\color{green} \textbf{[philip]: #1}}}
\newcommand{\pohan}[1] {{\color{blue} \textbf{[Po-han]: #1}}}
\newcommand{\sravan}[1]{{\color{red} \textbf{[Sravan]: #1}}}

\newcommand{\N}{N}
\newcommand{\X}[1]{X_{#1}}
\newcommand{\Xcom}{\tilde{X}_2}
\newcommand{\Y}{Y}
\newcommand{\Z}[1]{Z_{#1}}
\newcommand{\Zhat}[1]{\hat{Z}_{#1}}
\newcommand{\Xhat}[1]{\hat{X}_{#1}}
\newcommand{\Yhat}{\hat{Y}}
\newcommand{\xdim}[1]{n_{#1}}
\newcommand{\ydim}{p}
\newcommand{\zdim}[1]{m_{#1}}
\newcommand{\task}{\Phi}
\newcommand{\enc}[1]{E_{#1}}
\newcommand{\dec}{D}
\newcommand{\taskdec}{M}
\newcommand{\nenc}[1]{\mathrm{Enc}_{#1}}
\newcommand{\ndec}{\mathrm{Dec}}

\newcommand\normx[1]{\Vert#1\Vert}

%% file: section/abstract.tex
Efficient compression of correlated data is essential to minimize communication overload in multi-sensor networks. 
In such networks, each sensor independently compresses the data and transmits them to a central node.
A decoder at the central node decompresses and passes the data to a pre-trained machine learning-based task model to generate the final output. 
Due to limited communication bandwidth, it is important for the compressor to learn only the features that are relevant to the task.
Additionally, the final performance depends heavily on the total available bandwidth. In practice, it is common to encounter varying availability in bandwidth. Since higher bandwidth results in better performance, it is essential for the compressor to dynamically take advantage of the maximum available bandwidth at any instant.
In this work, we propose a novel distributed compression framework composed of independent encoders and a joint decoder, which we call neural distributed principal component analysis (NDPCA). 
NDPCA flexibly compresses data from multiple sources to any available bandwidth with a single model, reducing compute and storage overhead. 
NDPCA achieves this by learning low-rank task representations and efficiently distributing bandwidth among sensors, thus providing a graceful trade-off between performance and bandwidth. 
Experiments show that NDPCA improves the success rate of multi-view robotic arm manipulation by $9\%$ and the accuracy of object detection tasks on satellite imagery by $14\%$ compared to an autoencoder with uniform bandwidth allocation.
\footnote{\href{https://github.com/UTAustin-SwarmLab/Task-aware-Distributed-Source-Coding}{ https://github.com/UTAustin-SwarmLab/Task-aware-Distributed-Source-Coding}.}

%% file: section/intro_v0.tex



Efficient data compression plays a pivotal role in multi-sensor networks to minimize communication overload. 
Due to the limited communication bandwidth of such networks, it is often impractical to transmit all sensor data to a central server, and compression of the data is necessary. It is important for the sensors to compress the respective data independently, to avoid communication overload in the network. Information theory literature refers to this setting as distributed source coding~\cite{networkinfotheory}, where the goal is to recover the original data with minimal distortion. In many cases, the data collected by the sensors is only processed by a downstream task model, \textit{e.g.}, an object detection model, but not by humans, and hence the original distributed source coding goal of minimizing reconstruction error is no longer applicable. Instead, the goal should be to maximize the performance of the downstream task model. Additionally, in practice, data collected by multi-sensor networks is often correlated \textit{e.g.} stereo cameras with overlapping fields of view. To improve communication efficiency, it is important for the compression framework to take advantage of the correlation and avoid transmission of redundant data. Combining both objectives, the final goal of the distributed compression framework is to learn relevant features that maximize the task performance, while avoiding the transmission of redundant features by exploiting the correlation between sources. Together, we refer to the distributed compression of task-relevant features as \textit{task-aware distributed source coding}.

However, existing compression methods fail to combine the following three aspects: 
1. Existing distributed compression methods perform poorly in the presence of a task model. 
Although neural networks have been shown to be capable of compressing stereo images \citep{balle2016end,balle2018variational} and correlated images \citep{zhang2023ldmic}, existing methods focus on reconstructing image data, but not for downstream tasks.
2. Existing task-aware compression methods cannot take advantage of the correlation of sources. Previous works only consider compressing task-relevant features of single source \cite{cheng2022taskaware,TaskDependentCodebook,TaskAwareJPEGImage,nakanoya2023co,cheng2021data}, but not multiple correlated sources. 
3. All existing methods for 1 \& 2, especially those based on neural networks, only compress data to a fixed level of compression but not to multiple levels.
Thus, they cannot operate in environments with different demands of compression levels and require a separate model trained for each compression level.
Here, we note that we use the term bandwidth to indicate the information bottleneck in the dimension of transmitted data. Based on the choice of quantization, it is straightforward to convert the latent dimension to other popular metrics such as bits per pixel (bpp) in the case of image sources. Additionally, we consider the scenario of total bandwidth constraint for the uplink, which is typical for wireless  networks~\cite{liu2022training}.

We design neural distributed principal component analysis (NDPCA)–a distributed compression framework that can transmit task-relevant features at multiple compression levels.
We consider the case where a task model at the central node requires data from all sources and the bandwidth in the network is not consistent over time, as shown in Fig. \ref{fig:system_graph}. 
In NDPCA, neural encoders $\enc{1}, \enc{2}, \dots, \enc{K}$ first independently compress correlated data $\X{1}, \X{2}, \dots, \X{K}$ to latent representations $\Z{1}, \Z{2}, \dots, \Z{K}$. A proposed module called distributed principal component analysis (DPCA) further compresses these representations to any lower dimension according to the current bandwidth and decompresses the data at the central node. 
Finally, a neural decoder at the central node decodes the representations $\Zhat{1}, \Zhat{2}, \dots, \Zhat{k}$ to $\Xhat{1}, \Xhat{2}, \dots,\Xhat{K}$ and feeds them into a task. 
Task-aware compression aims to minimize task loss, defined as the difference in task outputs with and without compression, such as the difference in object detection results. 
Due to the significant training cost involved, we avoid training the task model, which is usually a large pre-trained neural network.

To highlight our proposed method, NDPCA learns task-relevant representations with a single model at multiple compression levels.
It includes a neural autoencoder to generate uncorrelated task-relevant representations in a fixed dimension. It is desirable to learn uncorrelated representations to prevent the transmission of redundant information. It also includes a module for linear projection, DPCA, to allocate the available bandwidth among sources based on the importance of the task, by observing the respective principal components, and then further compressing the representations to any desired dimension. 
By harmoniously combining the neural autoencoder and the linear DPCA module, NDPCA generates representations that are more compressible in limited bandwidths, providing a graceful trade-off between performance and bandwidth. 

\input{figure_latex/system_model}

\textbf{Contributions:} Our contributions are three-fold: 
First, we formulate a task-aware distributed source coding problem that optimizes the compression for a given task instead of reconstructing the sources (\secref{sec:formulation}). 
Second, we provide a theoretical justification for the framework by analyzing the case of a linear compressor and a linear task (\secref{sec:DPCA}).
Finally, we propose a task-aware distributed source coding framework, NDPCA, that learns a single model for different levels of compression to handle any type of source and task(\secref{sec:NDPCA}). 
We validate NDPCA with tasks of CIFAR-10 image denoising, multi-view robotic arm manipulation, and object detection of satellite imagery (\secref{sec:exps}). NDPCA results in a $1.2$dB increase in PSNR, a $9\%$ increase in success rate, and a $14\%$ increase in accuracy compared to an autoencoder with uniform bandwidth allocation, for the respective experiments mentioned above.

%% file: figure_latex/system_model.tex

\begin{figure}[t!]
\centering
\includegraphics[width=0.85\textwidth]{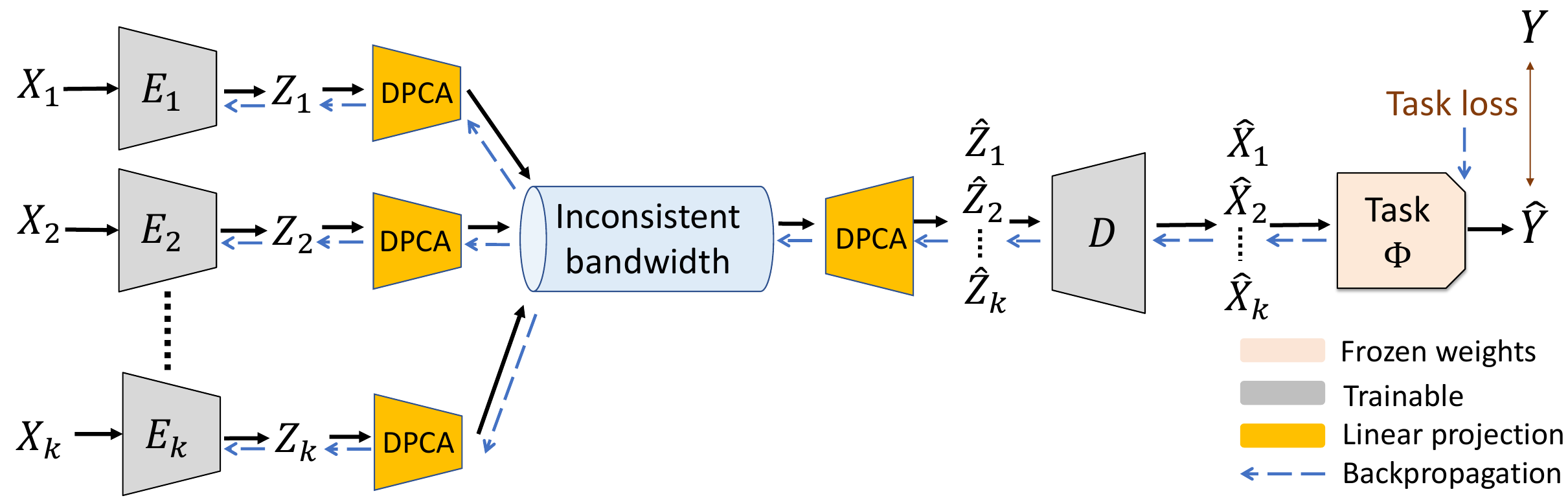}
\caption{\small
    \textbf{Task-aware distributed source coding with NDPCA.}
    $\X{1}, \dots, \X{k}$ are correlated data sources. Encoders $\enc{1}, \dots, \enc{k}$ independently compress data to latent representations $\Z{1}, \dots, \Z{k}$.
    Using linear matrices, the DPCA module projects the representations to any lower dimension at the encoder and projects them back to the original data space at the decoder, which allocates the bandwidth of sources based on the importance of the task $\task$. The goal is to find the optimal encoders and decoder that minimize the final task loss.}
\label{fig:system_graph}
\end{figure}


%% file: section/problem_formulation.tex
\label{sec:formulation}
We now define the problem statement more formally. Consider a set of $K$ correlated sources. Let $x_i \in \mathbb{R}^{\xdim{i}}$ denote the sample from source $i$ where $i \in \{ 1,2, \dots, K \}$. Samples from each source $i$ are compressed independently by encoder $E_i$ to a latent representation $z_i \in \mathbb{R}^{\zdim{i}}$ such that $\sum_{i=1}^{K}m_i = m$, where $m$ is the total bandwidth available. A joint decoder $D$ receives the representations  $\{z_1, z_2, \dots, z_k \}$ and reconstructs the sources $\{ \hat{x}_1, \hat{x}_2, \dots, \hat{x}_k \}$. 
In the setting without a task, the goal is to find a set of encoders and a decoder to recover the inputs $\{x_1, x_2, \dots, x_k\}$ with minimal loss:
\begin{equation}
\begin{aligned}
   \operatorname*{argmin}_{ E_1, E_2, \dots, E_k, D} \sum_{i=1}^{K} \mathcal{L}_{\text{rec}}(x_i, \hat{x}_i) \quad \text{(\textit{Task-agnostic} distributed source coding)},
\label{eq:rec_loss}
\end{aligned}
\end{equation}
where $\mathcal{L}_{\text{rec}}$ is the reconstruction loss, \textit{e.g.}, the mean-squared error loss.

In the presence of a task $\Phi$, it takes the reconstructed inputs to compute the final output $\Phi(\hat{x}_1, \hat{x}_2, \dots, \hat{x}_k)$. The goal is to find a set of encoders and a decoder such that the task loss $\mathcal{L}_{\text{task}}$ is minimized, where $\Phi(x_1, x_2, \dots, x_k)$ is the task output computed without compression. We refer to this problem as \textit{task-aware} distributed source coding, which is the main focus of this paper: 
\begin{equation}
\begin{aligned}
   \operatorname*{argmin}_{ E_1, E_2, \dots, E_k, D}\mathcal{L}_{\text{task}}( \Phi(x_1, x_2, \dots, x_k) , \Phi(D(E_1(x_1),E_2(x_2), \dots, E_k(x_k) ))) \\
   \quad \text{(\textit{Task-aware} distributed source coding)},
\label{eq:task_loss}
\end{aligned}
\end{equation}
where $\mathcal{L}_{\text{task}}$ is the task loss, \textit{e.g.}, the difference of bounding boxes when the task is object detection.

\textbf{Bandwidth allocation:}
In the previous formulations, we assume that the output dimensions of encoders are known a priori. 
However, the dimensions are related to the compression quality of each encoder, which is also a design factor. 
That is, given the total available bandwidth $\zdim{}$, we first need to obtain the optimal $\zdim{i}$ for each source $i$, then, we can design the optimal encoders and decoder accordingly. 
Finding the optimal set of bandwidths for a given task is a long-standing open problem in information theory \cite{el2011network}, even for the simple task of a modulo-two sum of two binary sources \cite{modulotwo}. Also, existing works  \cite{zhang2023ldmic,whang2021neuralDSC,NDIC} largely assume a fixed latent dimension for sources and train different models for different total available bandwidth $\zdim{}$, which is, of course, suboptimal.
In this paper, our framework provides heuristics to the underlying key challenge of optimally allocating available bandwidth, \textit{i.e.}, deciding $\zdim{i}$, while adapting to different total bandwidths $\zdim{}$ with a single model.

%% file: section/DPCA.tex
\label{sec:DPCA}



We start with a motivating example of task-aware distributed source coding under the constraint of linear encoders, a decoder, and a linear task. 
We first solve the linear setting using our proposed method, distributed principal component analysis (DPCA). 
We then describe how DPCA compresses data to different bandwidths and analyze the performance of DPCA. In this way, we gain insights into combining DPCA with neural autoencoders in later \secref{sec:NDPCA}. 



\textbf{DPCA Formulation:} 
We consider a linear task for two sources, defined by the task matrix $\Phi\in \mathbb{R}^{\ydim\times(\xdim{1}+\xdim{2})}$, where the sources $x_1 \in \mathbb{R}^{n_1}$ and $x_2 \in \mathbb{R}^{n_2}$ are of dimensions $n_1$ and $n_2$, respectively, and the task output is given by $y = \Phi x \in \mathbb{R}^{\ydim}$, where $x = [x_1^\top, x_2^\top]^\top$. Without loss of generality, we assume the sources to be zero-mean. 
Now, we have $N$ observations of 
two sources $\X{1}\in \mathbb{R}^{\xdim{1}\times\N}$ and $\X{2}\in \mathbb{R}^{\xdim{2}\times\N}$ and their corresponding task outputs $\Y = \Phi(\X{}) \in \mathbb{R}^{\ydim\times\N}$, where $X=[X_1^\top X_2^\top]^\top$. 

We aim to design the optimal linear encoding matrices (encoders) $E_1 \in \mathbb{R}^{\zdim{1}\times\xdim{1}}$, $E_2 \in \mathbb{R}^{\zdim{2}\times\xdim{2}}$, and the decoding matrix (decoder) $D \in \mathbb{R}^{(\xdim{1}+\xdim{2})\times(\zdim{1}+\zdim{2})}$ that minimizes the task loss defined as the Frobenius norm of $\Phi(X)-\Phi(\hat{X})$, where $\hat{X}$ is the reconstructed $X$. 
We only consider the non-trivial case where the total bandwidth is less than the task dimension, $m = m_1 + m_2 < \ydim$, i.e., the encoders cannot directly calculate the task output locally and transmit it to the decoder. For now, we assume that $m_1$ and $m_2$ are given, and we discuss the optimal allocation later in this section. 

Letting $\Z{1} = E_1 X_1\in\mathbb{R}^{\zdim{1}\times\N}$ and $\Z{2} = E_2 X_2\in\mathbb{R}^{\zdim{2}\times\N}$ denote the encoded representations and $M = \Phi D$ denote the product of the task and decoder matrices, we solve the optimization problem: 
\begin{subequations}
\begin{align}
\enc{1}^*, \enc{2}^*, \taskdec^* =
\operatorname*{argmin}_{\enc{1}, \enc{2}, \taskdec} & \quad \|Y-\taskdec \Z{}\|_2^2 \label{eq:DPCA_obj} \\ 
\mathrm{s.t. }  
& \quad \Z{} = \begin{bmatrix}\Z{1}\\ \Z{2}\end{bmatrix} = \begin{bmatrix}\enc{1}\X{1} \\ \enc{2}\Xcom \end{bmatrix}, \label{eq:DPCA_enc} \\
& \quad \Z{} \Z{}^\top  = \mathbb{I}_{\zdim{}} \label{eq:DPCA_orth}, \\
& \quad \Yhat = \task \dec \Z{}= {\taskdec} \Z{}, ~~~\Y = \task 
\label{eq:DPCA_dec}
\begin{bmatrix}\X{1} \\ \X{2} \end{bmatrix}.
\end{align}
\label{eq:DPCA}
\end{subequations}
Note that solving $\taskdec$ is identical to solving the decoder $\dec$ since we can always convert $M$ to $\dec$ by the generalized inverse of task $\task$. 
The encoders $E_1$ and $E_2$ project the data to representations $\Z{1}$ and $\Z{2}$ in \Eqref{eq:DPCA_enc}. 
We constrain the representations to be orthonormal vectors in \Eqref{eq:DPCA_orth} as in the normalization in principal component analysis (PCA) for the compression of a single source \citep{PCA}. This constraint lets us decouple the problem into subproblems later in \Eqref{eq:decouple}. 
Finally, in \Eqref{eq:DPCA_dec}, the decoder $\dec$ decodes $\Z{1}$ and $\Z{2}$ to $\Xhat{1}$ and $\Xhat{2}$ and passes the reconstructed data to task $\task$. 


\textbf{Solution:} We now solve the optimization problem in~\eqref{eq:DPCA}. 
\label{sec:sub_DPCA}
For any given $E_1,E_2$ (thus, a given $Z$), we can optimally obtain $\taskdec^*= \Y\Z{}^\top(\Z{}\Z{}^\top)^{-1} = \Y\Z{}^\top$ by linear regression.
Now, we are left to find the optimal encoders $E_1,E_2$.
First, a preprocessing step removes the correlation part of $\X{1}$ from $\X{2}$ by subtracting the least-square estimator $\hat{\X{2}}(\X{1})$: 
\begin{equation}
\Xcom = \X{2} - \hat{\X{2}}(\X{1}) = \X{2} - \X{2}\X{1}^\top(\X{1}\X{1}^\top)^{-1}\X{1}.
\label{eq:preprocess}
\end{equation}
The orthogonality principle of least-square estimators \citep[p.386]{FundamentalsSignalProcessing} ensures that $\X{1}{\Xcom}^\top=\mathbf{0}_{\xdim{1}\times \xdim{2}}$.
We decouple the objective in \Eqref{eq:DPCA_obj} with respect to $\enc{1},\enc{2}$ by the orthogonality principle and \Eqref{eq:DPCA_orth}:
\begin{equation}
\begin{aligned}
\operatorname*{min}_{E1, E2} \|\Y-\taskdec^* \Z{}\|_2^2 
= \|\Y\|_2^2 - \max_{\enc{1}, \enc{2}}\|\taskdec^*\|_2^2 
= \|\Y\|_2^2 - \operatorname*{max}_{\enc{1}} \|\Y_1\X{1}^\top \enc{1}^\top\|_2^2 - \operatorname*{max}_{\enc{2}} \|\Y_2 \Xcom^\top \enc{2}^\top\|_2^2,
\end{aligned}
\label{eq:decouple}
\end{equation}
where $\Y = \task \X{} = \begin{bmatrix} \task_{1} \task_{2}\end{bmatrix} \begin{bmatrix}\X{1}^\top \X{2}^\top\end{bmatrix}^\top= \Y_{1} + \Y_{2}$.
We then have two subproblems from \Eqref{eq:DPCA}: \\
\begin{minipage}[t]{0.5\textwidth}
\begin{equation}
\begin{aligned}
\enc{1}^*=\operatorname*{argmax}_{\enc{1}} & \quad \|\task_1\X{1} \X{1}^\top\enc{1}^\top\|_2^2 \\
\mathrm{s.t. }  
& \quad \enc{1}\X{1}\X{1}^\top\enc{1}^\top = \mathbb{I}_{\zdim{1}},
\label{eq:DPCA_sub1}
\end{aligned}
\end{equation}
\end{minipage}\begin{minipage}[t]{0.5\textwidth}
\begin{equation}
\begin{aligned}
\enc{2}^*=\operatorname*{argmax}_{\enc{2}} & \quad \|\task_2\Xcom\Xcom^\top\enc{2}^\top\|_2^2 \\
\mathrm{s.t. }  
& \quad \enc{2}\Xcom{\Xcom}^\top\enc{2}^\top = \mathbb{I}_{\zdim{2}}. 
\label{eq:DPCA_sub2}
\end{aligned}
\end{equation}
\end{minipage}

The two subproblems are the canonical correlation analysis \citep{CCA}, which can be solved by whitening $\enc{1}\X{1},\enc{2}\Xcom$ and singular value decomposition (see \cite{CCA} for details). 

\textbf{Dynamic bandwidth:} So far, we solved the case for fixed bandwidths $\zdim{1}$ and $\zdim{2}$. 
We now describe ways to determine the optimal bandwidth allocation given a current total bandwidth $\zdim{}$. 
To do so, DPCA solves \Eqref{eq:DPCA_sub1} and \Eqref{eq:DPCA_sub2} with $\zdim{1}=\xdim{1}$ and $\zdim{2}=\xdim{2}$ and obtains $\enc{1}^*$, $\enc{2}^*$ and all pairs of canonical directions and correlations.
Canonical directions and correlations can be analogized to a more general case of singular vectors and values.
Similar to PCA, the sums of squares of canonical correlations are the optimal values of \Eqref{eq:DPCA_sub1} and \Eqref{eq:DPCA_sub2}, so DPCA sorts all the canonical correlations in descending order and chooses the first $\zdim{}$ pairs of canonical correlations and directions. 
These canonical correlations determine the optimal encoders $\enc{1}^*,\enc{2}^*$ and decoder $\dec^*$, which indirectly solves $\zdim{1}$ and $\zdim{2}$. Intuitively, the canonical correlations indicate the importance of a direction to the task, and we prioritize the transmission of directions by importance. For simplicity, we only consider the case of $2$ sources. DPCA can easily compress more sources simply by constraining all $\Z{}$s to be independent and thus decoupling the original problem \Eqref{eq:DPCA} to more subproblems.

\textbf{Performance analysis of DPCA:}
When DPCA compresses new data matrices with encoder $\enc{1}^*$ and $\enc{2}^*$, the preprocessing step is invalid as the encoders cannot communicate with each other. 
So for DPCA to perform optimally while skipping the step, the two data matrices need to be uncorrelated, namely, $\hat{\X{2}}(\X{1})=0$, because in such case the preprocessing step removes nothing from the data sources. 
Given that correlated sources lead to suboptimality of DPCA, we characterize the performance between the joint compression, PCA, and the distributed compression, DPCA, under the same bandwidth in \lmaref{lemma:boundDPCA} with the simplest case of reconstruction, namely, $\task=\mathbb{I}_{\ydim}$. 
In this setting, the canonical correlation analysis is relaxed to the singular value decomposition, which is later used for NDPCA in \secref{sec:NDPCA}. 
\begin{lemma}[Bounds of DPCA Reconstruction]
\label{lemma:boundDPCA}
Given a zero-mean data matrix and its covariance, 
$$X = \begin{bmatrix}
X_1 \\
X_2
\end{bmatrix} \in \mathbb{R}^{(\xdim{1}+\xdim{2})\times N}, 
XX^\top =  \underbrace{\begin{bmatrix}
\mathrm{Cov}_{11} & \mathbf{0} \\
\mathbf{0} & \mathrm{Cov}_{22}
\end{bmatrix}}_{X_{\mathrm{diag}}}
+ \underbrace{\begin{bmatrix}
\mathbf{0} & \mathrm{Cov}_{12} \\
\mathrm{Cov}_{21} & \mathbf{0} 
\end{bmatrix}}_{\Delta X}, $$
assume that $\Delta X$ is relatively smaller than $XX^\top$, and $XX^\top$ is positive definite with distinct eigenvalues.
For PCA's encoding and decoding matrices $E_{\mathrm{PCA}}, D_{\mathrm{PCA}}$ and DPCA's encoding and decoding matrices $E_{\mathrm{DPCA}}, D_{\mathrm{DPCA}}$, the difference of the reconstruction losses is bounded by 
\begin{align*}
0 \leq 
\| X - D_{\mathrm{DPCA}}\ E_{\mathrm{DPCA}}(X)\|_2^2 - \| X - D_{\mathrm{PCA}} E_{\mathrm{PCA}}(X)\|_2^2
= - \sum_{i=\zdim{}+1}^{\xdim{1}+\xdim{2}} \lambda_ie_i^\top \Delta X e_i.
\end{align*}
where $\lambda_i$ and $e_i$ are the $i$-th largest eigenvalue and eigenvector of $XX^\top$, $\mathrm{Tr}$ is the trace function, and $\zdim{}$ is the dimension of the compression bottleneck. 
\end{lemma}
The proof of \lmaref{lemma:boundDPCA} is in \appref{sec:app_proof}. Note that $\Delta X$ is the correlation of sources, so as $\|\Delta X\|_F$ gets smaller, the difference of PCA and DPCA is closer to $0$. That is, as the covariance decreases, DPCA performs more closely to PCA, which is the optimal joint compression. 

To summarize, 
uncorrelated data matrices $\X{1}, \dots, \X{K}$ are desired for DPCA. 
If so, DPCA 
optimally decides the bandwidths of all sources based on the canonical correlations, representing their importance for the task. 
One application of DPCA is that encoders can use the remaining unselected canonical directions to improve compression when the available bandwidth is higher later.




%% file: section/distributed_autoencoder.tex
\label{sec:NDPCA}

The theoretical analysis in the previous section indicates that DPCA has two drawbacks: it only compresses data optimally if sources are uncorrelated, and it only works for linear tasks. However, DPCA dynamically allocates bandwidth to sources based on their importance. 
On the other hand, neural autoencoders are shown to be powerful tools for compressing data to a fixed dimension but cannot dynamically allocate bandwidth. 
This contrast motivates us to harmoniously combine a neural autoencoder to generate representations and then pass them through DPCA to compress and find the bandwidth allocation. We refer to the combination of a neural autoencoder and DPCA as neural distributed principal component analysis (NDPCA). With a single neural autoencoder and a matrix at each encoder and decoder, NDPCA adapts to any available bandwidth and flexibly allocates bandwidth to sources according to their importance to the task.


\textbf{Outline:} 
NDPCA has two encoding stages, as shown in \figref{fig:system_graph}: First, the neural encoder at each $k$-th source encodes data $\X{k}$ to a {fixed-dimensional} representation $\Z{k}$ for $k \in [K]$. Then the DPCA linear encoder {adapts} the dimension of $\Z{k}$  
via linear projection  according to the available bandwidth and the correlation among the sources as per \Eqref{eq:DPCA_sub1}.  
Similarly, the decoding of NDPCA is also performed in two stages. First, 
the DPCA linear decoder reconstructs the $K$ fixed-dimensional representations $\Zhat{1}, \Zhat{2}, \dots, \Zhat{K}$, based on which the joint neural decoder generates the estimate of data $\Xhat{1}, \Xhat{2}\, \dots, \Xhat{K}$. These estimates are then passed to the neural task model $\task$ to obtain the final task output $\Yhat$. Note that since we have a non-linear task model here, we envision that the neural encoders generate non-linear embedding of the sources, while the DPCA mainly adapts the dimension appropriately as needed; the role of the DPCA here is to reliability reconstruct the embedding $\hat{Z}$s, which corresponds to the case described in \lmaref{lemma:boundDPCA} with the task matrix $\Phi$ as identity.


\textbf{Training procedure:}
During the training of NDPCA, the weights of the task are always frozen because it is usually a large-scale pre-trained model that is expensive to re-train. 
We aim to learn the $K$ neural encoders and the joint neural decoder which minimize the loss function: 
\begin{equation}
\mathcal{L}_{\mathrm{tot}} = 
\lambda_{\mathrm{task}} \underbrace{\|\Yhat - \Y \|_F^2}_{\mathrm{task~loss}}
+ \lambda_{\mathrm{rec}} \underbrace{\left(\|\Xhat{1} - \X{1}\|_F^2 + \|\Xhat{2} - \X{2}\|_F^2 + \dots \|\Xhat{K} - \X{K}\|_F^2 \right)}_{\mathrm{reconstruction~loss}}.
\label{eq:overall_loss}
\end{equation}
In the task-aware setting when $\lambda_{\mathrm{rec}}=0$, the neural autoencoder fully restores task-relevant features, which is the main focus of this paper.
When $\lambda_{\mathrm{task}}=0$, the neural autoencoder learns to reconstruct the data $\X{}$, which is the task-agnostic setting later compared in \secref{sec:exps}. 

We now discuss how to encourage NDPCA to work well under \textit{various available bandwidths} with DPCA during the training phase. 
We begin by making observations on the desired property of the neural embeddings arising from the limitations of the DPCA: (1) uncorrelatedness: \lmaref{lemma:boundDPCA} shows that DPCA is more efficient when the correlation among the intermediate representations is less. (2) linear compressibility: we encourage the neural autoencoder to generate low-rank representations, which can be compressed by only a few singular vectors, making them more bandwidth efficient.  

We tried to explicitly encourage the desired properties with additional terms in \Eqref{eq:overall_loss}, but they all adversely affect the task performance. 
To obtain uncorrelated representations, we tried penalizing the cosine similarity between the representations. We also tried similar losses that penalize correlation, as per \cite{VICREG,DecoupleAndSample,bousmalis2016DSN,DisentanglementVAE}, but none improves the task performance.
We observed that the autoencoder automatically learns representations with small correlation, and any explicit imposition of complete uncorrelatedness is too strong. 
For linear compressibility, we tried penalizing the convex low-rank approximation–the nuclear norm–of the representations, as per \cite{FactorizedOrthLatentSpaces,fazel2002matrix}. However, we observe a similar trend in the final task performance as the network tends to minimize the nuclear norm while harming the task performance. For the comparison of the resulting performance, see \appref{sec:ablation_norm}.



In this regard, we propose a novel linear compression module that allows us to adapt to DPCA during training rather than using additional terms in the loss.
We introduce a \textit{random-dimension} DPCA projection module to improve performance in lower bandwidths. 
It projects representations $\Z{}$ to a low dimension randomly chosen, simulating projections in various available bandwidths during inference. 
It can be interpreted as a differentiable singular value decomposition with a random dimension, described in \algoref{alg:pca_random_projection}. 
For encoding, it first normalizes the representations and performs singular value decomposition on all sources. Then, it sorts the vectors by the singular values and randomly selects the number of vectors to use for projection. 
For decoding, it decodes with the selected singular vectors again and denormalizes the data. 
Note that during training, we only run \algoref{alg:pca_random_projection} on a batch. 
This module helps to improve the overall performance over a range of bandwidths, and we show the ablation study of this module in \appref{sec:ablation_DPCA}.

\textbf{Inference:} 
With the training data, the DPCA projection module first saves the mean of representations $\Z{}$ and the encoder and decoder matrices in the maximum bandwidth.
It only needs to save for the maximum bandwidth because its rows and columns are already sorted by the singular values, which represent the importance of each corresponding vector.
During inference, when the current bandwidth is $\zdim{}$, it chooses the top $\zdim{}$ rows and columns of the saved encoders and decoder matrices to encode and decode representations. 
No retraining is needed for different bandwidths. Only the storage of a neural autoencoder and a linear matrix at each encoder and decoder is needed. 

\textbf{Robust task model:}
We pre-train the task model with randomly cropped and augmented images to make the model less sensitive to noise in the input image space, namely, the model has a smaller Lipschitz constant. This augmentation trick is based on \cite{nakanoya2023co}. 
A robust task model has a smaller Lipschitz constant, so it is less sensitive to the input noise injected by decompression when we concatenate it with the neural autoencoder. 
For a detailed analysis of the performance bounds between robust task and task-aware autoencoders, see \appref{sec:robust_aware}.

\input{figure_latex/DPCA_alg}




%% file: figure_latex/DPCA_alg.tex
\begin{wrapfigure}{R}{0.7\textwidth}
\begin{minipage}{0.7\textwidth}
\vspace{-2.5em}
\setlength{\textfloatsep}{4em}
\begin{algorithm}[H]
  \footnotesize
  \caption{\small{Projection into a random low dimension using DPCA}}
  \label{alg:pca_random_projection}
\begin{algorithmic}[1]
    \State {\bfseries Input:} A size $b$ batch of latent representations $\Z{i}\in \mathbb{R}^{b\times\zdim{i}}$ from each source $i$, min and max bandwidth $\zdim{\text{min}}, \zdim{\text{max}}$
    \State {\bfseries Output:} Compressed representation $\Z{i}^{\zdim{}}$ of each source, reconstructed representation $\Zhat{}$ for all sources
    
    \Function{Encode}{$\Z{i}, \zdim{\text{min}}, \zdim{\text{max}}$}
    \For{each source $i$}
            \State $\bar{\Z{i}} \gets \Z{i} - \mathrm{Mean}(\Z{i})$ \Comment{Normalize representations}
            \State $s_i, V_i, H_i \gets \mathrm{SVD}( \bar{\Z{i}} )$ \Comment{Singular value decomposition}
        \EndFor
        \State $s ,V \gets \mathrm{Cat}(s_i)$, $\mathrm{Cat}(V_i)$ \Comment{Concatenate singular values and vectors}
        \State $\zdim{} \gets \mathrm{Rand}(\zdim{\text{min}}, \zdim{\text{max}})$  \Comment{Randomly choose projection dimension}
        \State $ s^{\zdim{}}, V^{\zdim{}} \gets \arg\max( [s, V], {\zdim{}})$
        \Comment{Select the top $\zdim{}$ values of $s$}
        \For{each source $i$}
            \State $V_i^{\zdim{}} \gets \{ V | V \in V^{\zdim{}}, V \in V_i \}$ \Comment{Select $\zdim{}$ vectors from sources}
            \State $\Z{i}^{\zdim{}} = \bar{\Z{i}}\times V_i^{\zdim{}}$ \Comment{Project $\Z{i}$ to lower dimensions}
        \EndFor
    \State \Return $\Z{\text{low}} \gets \mathrm{Cat}(\Z{i}^{\zdim{}})$ \Comment{Return Compressed representation }
    \EndFunction
    \Function{Decode}{$\Z{i}^{\zdim{}}$}
    \For{each source $i$}
    \State $\hat{\bar{\Z{i}}} \gets \Z{i}^{\zdim{}} \times \mathrm{Cat}(V_i^{\zdim{}})^\top$ \Comment{Decompressed representation}
    \State $\hat{Z{i}} \gets \hat{\bar{\Z{i}}} + \mathrm{Mean}(\Z{i})$  \Comment{Denormalize representations}
    \EndFor
    \State \Return $\hat{\Z{}} \gets \mathrm{Cat}(\hat{Z{i}})$ \Comment{Return reconstructed representations}
    \EndFunction
\end{algorithmic}
\end{algorithm}
\vspace{-2em}
\end{minipage}
\end{wrapfigure}

%% file: section/experimental_results.tex

\label{sec:exps}
We consider three different tasks to test our framework: ($a$) the denoising of CIFAR-10 images~\citep{CIFAR10}, ($b$) multi-view robotic arm manipulation~\citep{FERM}, which we refer to as the \textit{locate and lift} task, and ($c$) object detection on satellite imagery~\citep{airbusKaggle}. Across all the experiments, 
we assume that there are two data sources, referred to as \textit{views}, each containing partial information relevant to the task. 
We present our results based on the testing set and refer to our proposed method, task-aware NDPCA, as NDPCA for simplicity.
NDPCA includes a single autoencoder with a large dimension of representations $\Z{}\in \mathbb{R}^{2*\zdim{\mathrm{max}}}$. It then compresses representations and allocates bandwidth via DPCA, as discussed previously. 
We show that NDPCA can bridge the performance gap between distributed autoencoders and joint autoencoders, defined below, to allocate bandwidth and avoid transmitting task-irrelevant features. 
We also provide experiments of NDPCA with more than $2$ data sources in \appref{sec:4_sources} to demonstrate NDPCA's capability in more complicated settings. 

\textbf{Baselines:} 
We compare NDPCA against three major baselines. 
First is the task-aware joint autoencoder (JAE), where a single pair of encoder and decoder compresses both views. JAE is considered an upper bound of NDPCA since it can leverage the correlation between both views while avoiding encoding redundant information. Next is the task-aware vanilla distributed autoencoder (DAE), where two encoders independently encode the corresponding views to equal bandwidths and a joint decoder decodes the data. DAE is considered a lower bound of NDPCA since 
    both encoders utilize the same bandwidth regardless of the importance of the views for the task, while NDPCA allocates bandwidths in a task-aware manner. 
%
Last is the task-agnostic NDPCA which differs from NDPCA in the training loss of reconstructing the original views. Due to the novelty of the problem formulation, we cannot make a fair comparison with any of the existing approaches. For instance, ~\cite{balle2016end,balle2018variational,NDIC, zhang2023ldmic} focus purely on distributed compression of images for reconstruction and human perception, whereas ~\cite{TaskAwareJPEGImage,TaskDependentCodebook} focus on task-oriented compression but are limited for a single source. Additionally, none of the previous works consider datasets of unequal importance, again making any performance comparison unfair. Hence we focus mainly on an ablation study style of comparison of NDPCA, clearly highlighting and validating the advantages of our approach.

\input{figure_latex/image_sample}
\textbf{CIFAR-10 denoising:} We first consider a simple task of denoising CIFAR-10 images using two noisy observations of the same image, shown in \figref{fig:img_example} (a). We use CIFAR-10 as a toy example to clearly highlight the advantage of NDPCA in the presence of sources with unequal importance to the task. Due to the simplistic nature of the classification task, which only requires 4 bits (digit 0-9) as the information bottleneck, we choose denoising as our “task”, making it more suitable to showcase the performance across a range of available bandwidth. Here, the importance of each observation, or view, for the task is simply the noise level. 
For view 1, we consider an image corrupted with additive white Gaussian noise (AWGN) with a variance of ${0.1}^2$. And view 2 is highly corrupted by AWGN with a variance of $1$. All the images were normalized to $[0,1]$ before adding the noise. 
We compressed the noisy observations and passed the reconstructed images through a pre-trained denoising network. We then computed the final peak signal-to-noise ratio (PSNR) with respect to the clean image. 
Since the noise levels of both views are unequal, the importance of the task is unequal as well. The optimal bandwidth allocation should not be equal, thus showing the advantage of NDPCA.
Although view $1$ contains more information, not all bandwidth should be allocated to view $1$.
This problem is called the CEO problem \cite{CEO,RateCEO}.
In fact, even if one view is highly corrupted, we should still leverage that view and never allocate $0$ bandwidth to it.
We discuss why it is the case in \appref{sec:CEO}. 


\textbf{Locate and lift:} For the manipulation task, we consider a scenario in which a simulated $6$ degrees-of-freedom robotic arm controlled by a reinforcement learning agent inputs two camera views to locate and lift a yellow brick. 
We call the view from the robotic arm "arm-view" and the one recording the whole desk "side-view", as shown in \figref{fig:img_example} (b).
The two views are complementary to completing the task, details discussed in \appref{sec:ablation_lift}.
We trained the agent in a supervised-learning manner. We collected a dataset of observation and action pairs \cite{zhao2023learning} and trained an agent from the dataset. 
Then, we defined task loss as the $L_2$ norm of actions from images with and without compression and trained NDPCA to minimize the task loss through the agent. 
Literature calls this training method "behavior cloning" \cite{BehavioralCloning} as it learns from demonstrations.
Behavior cloning causes a drop in performance, but this paper only focuses on the performance degradation caused by compression, so we treat the behavior cloning agent with uncompressed views as the upper bound of our method. 

\textbf{Airbus detection:} This task considers using satellite imagery to locate Airbuses. Satellites observe overlapping images of an airport and transmit data to Earth through limited bandwidth, as shown in \figref{fig:img_example} (c). We crop all images in the dataset into smaller pieces ($224\times 224$ pixels). The two data sources are the upper $160$ pixels (source $1$) and the lower $104$ pixels of the image (source $2$) with $40$ pixels overlapped. 
Our object detection model follows the paper "You Only Look Once" (Yolo) \cite{Yolo}. 
The task loss here is the difference between object detection loss with and without compression. 

\input{figure_latex/results}

\input{figure_latex/tak-aware_airbus}
\textbf{Results:}
Our key results are: (1) Task-aware NDPCA outperforms task-agnostic NDPCA, and (2) bandwidth allocation should be related to the importance of the task. 
Across all experiments, shown in \figref{fig:results}(a)-(c), we see that task-aware NDPCA performs much better than task-agnostic NDPCA and DAE, which equally allocates bandwidths. 
We see from \figref{fig:results} that task-aware NDPCA provides a graceful performance degradation with respect to available bandwidth, with no additional training or storage of multiple models. 
On the other hand, DAE and JAE require retraining for every level of compression, so every sample point in the plot is a different model.

\figref{fig:results}(a) shows the results of denoising CIFAR-10 with NPDCA trained at $(\zdim{\mathrm{min}},\zdim{\mathrm{max}})=(8,64)$. 
Although view $1$ is more important than view $2$, DAE can only equally allocates bandwidth to both sources. NDPCA compresses the data and flexibly allocates bandwidths, as shown in \ref{fig:results}(d), where we can see that $\Z{1}$ has more bandwidth than $\Z{2}$. NDPCA results in $1.2$ dB gain in PSNR compared to DAE when $\zdim{}=64$.

\figref{fig:results}(b) shows the results of the locate and lift task with NPDCA trained at $(\zdim{\mathrm{min}},\zdim{\mathrm{max}})=(8,48)$. We set the length of an episode as $50$ time steps and measure the success rate in $100$ episodes. 
We show the upper bound, a behavior cloning agent without compression, in gray dotted lines.
The arm view is more important as it captures the precise location of the brick, and as expected, NDPCA allocates more bandwidth to the arm-view ($\Z{2}$), as seen in \figref{fig:results}(e). We see that NDPCA has a $9\%$ higher success rate compared to DAE when $\zdim{}=24$.

\figref{fig:results}(c) shows the results of the Airbus detection with NPDCA trained at $(\zdim{\mathrm{min}},\zdim{\mathrm{max}})=(8,40)$. We measured the mean average precision (mAP) with $40\%$ confidence score and $50\%$ intersection over the union as the thresholds. 
We show the uncompressed upper bound in gray dotted lines. NDPCA results in up to $14\%$ gain in mAP50 compared to DAE. 
In \figref{fig:results} (f), we plotted the ratio of the areas of both views, while equally splitting the overlapping part, in a dashed black line.
Surprisingly, NDPCA's empirical allocation of bandwidth is highly aligned with the theoretical ratio, supporting that it captures the importance of the task and allocates bandwidth according to it. 

\textbf{Comparison of NDPCA with JAE}: JAE uses the information from both views simultaneously to capture the best joint embedding for the task. In an ideal scenario, JAE will be the upper bound for the performance and hence easily performs better than DAE across all the experiments. Interestingly, in \figref{fig:results}(b) and (c), we see that NDPCA outperforms not only DAE but also JAE as well. We attribute it to the better representations present in higher-dimension latent space. It turns out that learning a high-dimensional representation and then projecting to a lower dimension space, like NDPCA, is more efficient compared to directly learning a low-dimensional representation, like JAE. This projection from higher dimensional to lower dimensional is similar to pruning large neural networks to identify effective sparse sub-networks.~\cite{frankle2018lottery,ye2020good}. We also note that Low-Rank Adaptation (LoRA)~\cite{hu2021lora} technique for large language models can be thought of as a similar approach.






\textbf{Task-aware v.s. task-agnostic:}
We plotted the reconstructed images of task-aware ($\lambda_{\mathrm{rec}}=0$) and task-agnostic ($\lambda_{\mathrm{task}}=0$) NDPCA in \figref{fig:task-compare}. Task-aware images are imperceptible to human eyes since they restore features of a non-linear task model, aligning with the results in \cite[Fig. 4]{nakanoya2023co}. 
For discussion of non-zero $\lambda_{\mathrm{task}}$ and $\lambda_{\mathrm{rec}}$, we refer readers to \appref{sec:weighted_task_loss}.

\textbf{Limitations:}
In general, autoencoders are poor at generalizing to out-of-distribution data and the drawback translates to NDPCA as well. When the testing set is noticeably different from the training set, the performance of NDPCA can get noticeably lower. Additionally, during training, DPCA performs the singular value decomposition in the training set. The decomposition operation can become ill-conditioned and unstable if the batch size is too small. An alternative approach could be a parametric low-rank decomposition such as LoRA~\cite{hu2021lora} or using adapter networks~\cite{houlsby2019parameter}, although the complexity increases and the compatibility with DPCA remains to be explored.


%% file: figure_latex/image_sample.tex
\begin{figure*}[t]
\centering
\begin{subfigure}{0.23\textwidth}
    \centering
    \captionsetup{width=1.1\linewidth}
    \includegraphics[width=\textwidth]{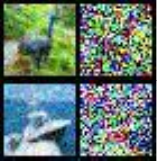}
    \caption{\small{CIFAR-10: view $1$ is less corrupted and thus contains more information about the original images.}}
\end{subfigure}
\hspace{4em}
\begin{subfigure}{0.23\textwidth}
    \centering
    \captionsetup{width=1.3\linewidth}
    \includegraphics[width=\textwidth]{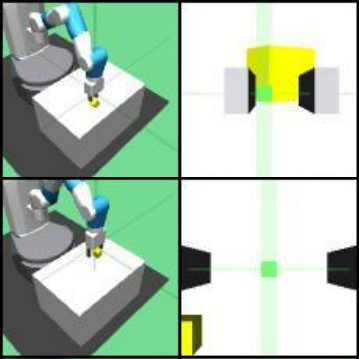}
    \caption{\small{Locate and lift: Side-view (column $1$) faintly captures the absolute position of objects, in contrast to the arm-view (column $2$).}}
\end{subfigure}
\hspace{4em}
\begin{subfigure}{0.23\textwidth}
    \centering
    \captionsetup{width=1.1\linewidth}
    \includegraphics[width=\textwidth]{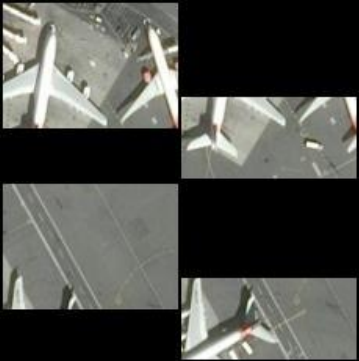}
    \caption{\small{Airbus detection: view $1$ and view $2$ observe different parts of the complete view with overlap.}}
\end{subfigure}
\caption{\small{\textbf{Datasets:}  (column $1$) view $1$.  (column $2$) view $2$. 
In all experiments, both views are correlated, but one view is more important than the other as it contains more information relevant to the task.}} 
\label{fig:img_example}
\vspace{-2em}
\end{figure*}

%% file: figure_latex/results.tex
\begin{figure*}[t]
\centering
\includegraphics[width=1\textwidth]{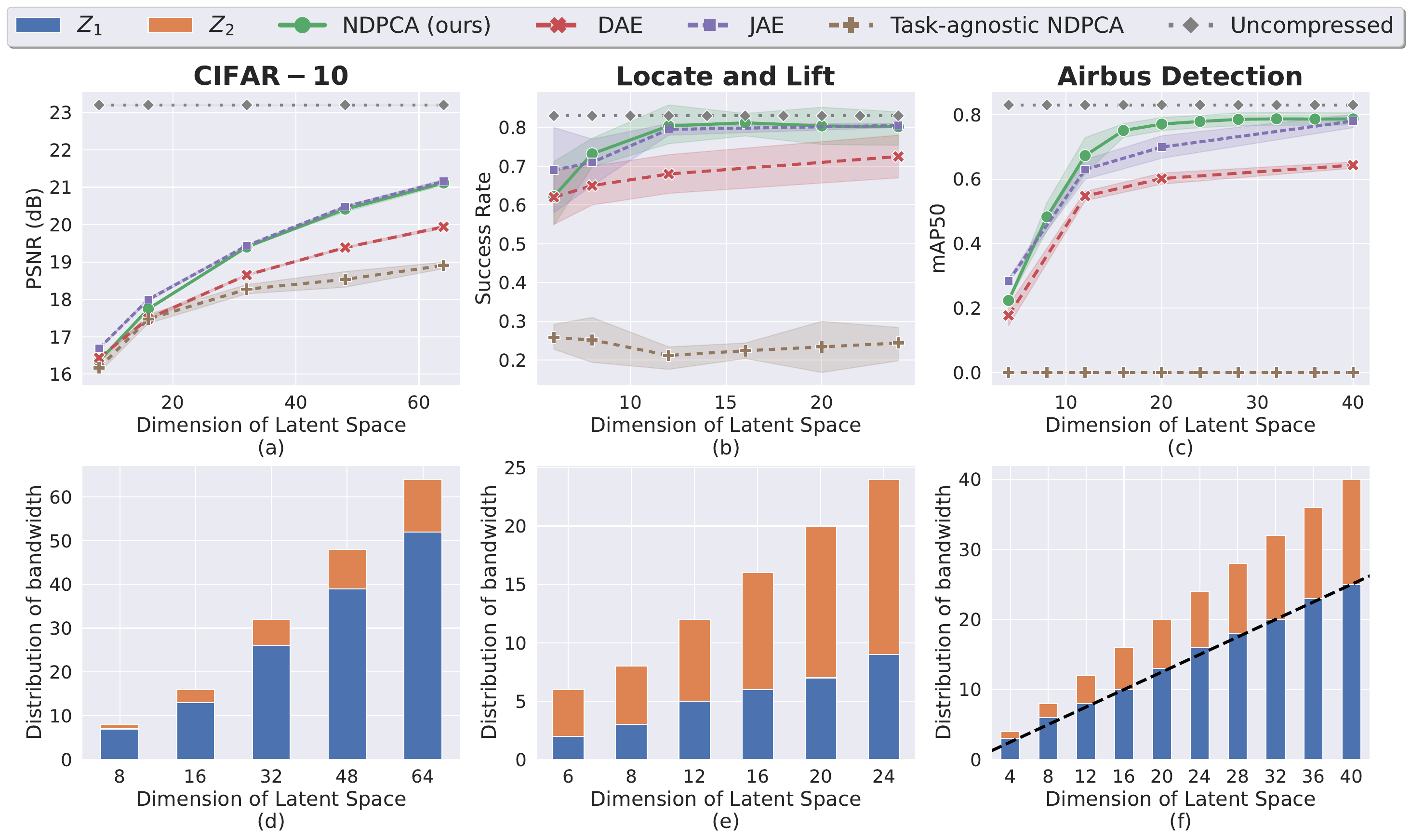}
\caption{\small{\textbf{Top:} Performance Comparison for 3 different tasks. Our method achieves equal or higher performance than other methods. \textbf{Bottom:} Distribution of total available bandwidth (latent space) among the two views for NDPCA (ours). The unequal allocation highlights the difference in the importance of the views for a given task.}}
\label{fig:results}
\vspace{-2em}
\end{figure*}


%% file: figure_latex/tak-aware_airbus.tex
\begin{wrapfigure}{R}{0.6\textwidth}
\vspace{-1em}
\begin{subfigure}{0.6\textwidth}
    \captionsetup{width=1.0\linewidth}
    \includegraphics[width=\textwidth]{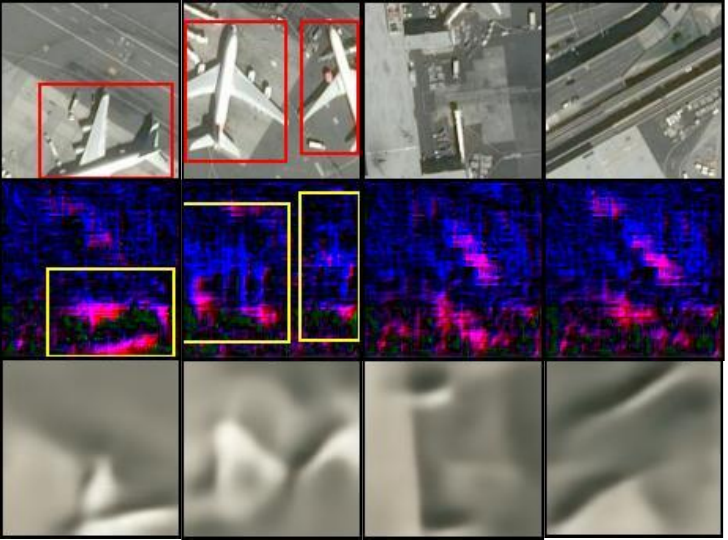}
\end{subfigure}
\caption{\small{\textbf{Task-aware v.s. task-agnostic:} 
     Ground-truth bounding boxes are red (row $1$), while detected boxes of task-aware are yellow (row $2$). Nothing is detected in the task-agnostic setting (row $3$). Task-agnostic images are perceptible to human eyes, while task-aware images capture task-relevant features, thus imperceptible to human eyes.}}
\label{fig:task-compare}
\vspace{-1em}
\end{wrapfigure} 

%% file: section/related_work.tex
\textbf{Information theoretic perspective:}
Slepian and Wolf \textit{et al.} are the first to obtain the minimum bandwidth of distributed sources to perfectly reconstruct data \citep{slepian_wolf}. However, they use exponentially complex compressors while assuming that the joint distribution of sources is known, which is impractical. 
In the presence of a task, finding the rate region of two binary sources has remained an open problem, even for modulo-two sum tasks \cite{modulotwo}. 
In terms of imperfectly reconstructing data with neural autoencoders, previous works consider compression of the original data to a fixed dimension \cite{whang2021neuralDSC,DRASIC}, while our work focuses on compressing data to any bandwidth with a task model. 

\textbf{Task-aware compression:} 
Real-world data, such as images or audio, are ubiquitous and high-dimensional, while downstream tasks that input the data only utilize certain features for the output.
Task-aware compression aims to compress data while maximizing the performance of a downstream task. Previous works analyze linear task \cite{cheng2022taskaware}, image compression \cite{TaskDependentCodebook, TaskAwareJPEGImage, nakanoya2023co,lossyCompressionLosslessPrediction}, future prediction \cite{cheng2021data}, and data privacy \cite{li2022differentially, TaskawarePrivacy}, while ours compresses distributed sources under limited bandwidth.

\textbf{Neural autoencoder:}
Previous works show the ability of neural autoencoders to generate meaningful and uncorrelated representations.
Instead of adding additional loss terms during training like \cite{VICREG, DecoupleAndSample, DisentanglementVAE, bousmalis2016DSN, CollaborativeEdge}, we use a random projection module to help a neural autoencoder learn uncorrelated and linear-compressible representations. 
Other works focus on designing new neural architectures for multi-view image compression \cite{zhang2023ldmic, NDIC}, while ours focuses on the framework to compress data to different compression levels. 
We choose autoencoders instead of variational autoencoders \cite{VAE, higgins2017betavae} because we focus on the compression of fixed representations rather than generative tasks from latent distributions. Also, autoencoders are more compatible with DPCA than variational autoencoders.

%% file: section/conclusion.tex
We proposed a theoretically grounded linear distributed compressor, DPCA, and analyzed its performance compared to the optimal joint compressor. 
Then, we designed a distributed compression framework called NDPCA by combining a neural autoencoder and DPCA to allocate bandwidth according to their importance to the task.
Experiments on CIFAR-10 denoising, locate and lift, and Airbus detection showed that NDPCA near-optimally outperforms task-agnostic or equal-bandwidth compression schemes. 
Moreover, NDPCA requires only one model and does not need to be retrained for different compression levels, which makes it suitable for settings with dynamic bandwidths. 


Avenues for future research include settings where the information flow is not unidirectional but bidirectional, such that the encoders and the decoder can communicate to compress data better. 
Discovering representations in a more complex space using kernel PCA instead of linear PCA and exploration of more complex non-linear correlations 
are also left as interesting 
future work. Another interesting direction to expand the work would be analyzing the robustness of the representations, both in the latent space with respect to corruption such as additive white Gaussian noise (AWGN) as well as with respect to the downstream task model. The current framework learns task-relevant features that are tied to the task model but the performance is expected to drop significantly when the task model is updated or changed. Hence, it is desirable to incorporate robust and transferable properties into the features learned.

%% file: section/acknowledgement.tex

This work was supported in part by the National Science Foundation 2148186, NASA 80NSSC21M0071, ARO Award W911NF2310062, ONR Award N00014-21-1-2379, NSF Award CNS-2008824, and Honda Research Institute through 6G$@$UT center within the Wireless Networking and Communications Group (WNCG) at the University of Texas at Austin. Any opinions, findings, and conclusions or recommendations expressed in this material are those of the authors and do not necessarily reflect the views of the National Science Foundation.

%% file: section/proof.tex
\section{Proofs of Lemmas}
\label{sec:aproof}
\subsection{Bounds of DPCA}
\label{sec:app_proof}
\begin{lemma-non}[Bounds of DPCA Reconstruction]
Given a zero-mean data matrix and its covariance, 
$$X = \begin{bmatrix}
X_1 \\
X_2
\end{bmatrix} \in \mathbb{R}^{(\xdim{1}+\xdim{2})\times N}, 
XX^\top =  \underbrace{\begin{bmatrix}
\mathrm{Cov}_{11} & \mathbf{0} \\
\mathbf{0} & \mathrm{Cov}_{22}
\end{bmatrix}}_{X_{\mathrm{diag}}}
+ \underbrace{\begin{bmatrix}
\mathbf{0} & \mathrm{Cov}_{12} \\
\mathrm{Cov}_{21} & \mathbf{0} 
\end{bmatrix}}_{\Delta X}, $$
assume that $\Delta X$ is relatively smaller than $XX^\top$, and $XX^\top$ is positive definite with distinct eigenvalues.
For PCA's encoding and decoding matrices $E_{\mathrm{PCA}}, D_{\mathrm{PCA}}$ and DPCA's encoding and decoding matrices $E_{\mathrm{DPCA}}, D_{\mathrm{DPCA}}$, the difference of the reconstruction losses is bounded by 
\begin{align*}
0 \leq 
\| X - D_{\mathrm{DPCA}}\ E_{\mathrm{DPCA}}(X)\|_2^2 - \| X - D_{\mathrm{PCA}} E_{\mathrm{PCA}}(X)\|_2^2
= - \sum_{i=\zdim{}+1}^{\xdim{1}+\xdim{2}} \lambda_ie_i^\top \Delta X e_i.
\end{align*}
where $\lambda_i$ and $e_i$ are the $i$-th largest eigenvalue and eigenvector of $XX^\top$, $\mathrm{Tr}$ is the trace function, and $\zdim{}$ is the dimension of the compression bottleneck. 
\end{lemma-non}
\begin{proof}
The lower bound is intuitive. We know that DPCA cannot outperform PCA since distributed coding cannot outperform joint coding and PCA is the optimal linear encoding. The reconstruction loss of PCA is always not greater than the loss of DPCA, thus the lower bound is $0$. 
Now consider the upper bound: 
\begin{align*}
& \; \| X - D_{\mathrm{DPCA}}E_{\mathrm{DPCA}}X\|_2^2 - \| X - D_{\mathrm{PCA}} E_{\mathrm{PCA}}X\|_2^2 \\
= & \;\mathrm{Tr}(XX^\top + D_{\mathrm{DPCA}}E_{\mathrm{DPCA}}X\left(D_{\mathrm{DPCA}}E_{\mathrm{DPCA}}X\right)^\top -2 D_{\mathrm{DPCA}}E_{\mathrm{DPCA}}XX^\top) \\
& \; - \sum_{i=\zdim{}+1}^{\xdim{1}+\xdim{2}} \lambda_i(XX^\top) \\
= & \; \mathrm{Tr}(X_{\mathrm{diag}} + \Delta X + D_{\mathrm{DPCA}}E_{\mathrm{DPCA}} X \left(D_{\mathrm{DPCA}}E_{\mathrm{DPCA}}X\right)^\top -2 D_{\mathrm{DPCA}}E_{\mathrm{DPCA}}XX^\top) \\
& \quad - \sum_{i=\zdim{}+1}^{\xdim{1}+\xdim{2}} \lambda_i(XX^\top) \\
= & \; \mathrm{Tr}(\Delta X + E_{\mathrm{DPCA}}^\top D_{\mathrm{DPCA}}^\top D_{\mathrm{DPCA}}E_{\mathrm{DPCA}} \Delta X -2 D_{\mathrm{DPCA}}E_{\mathrm{DPCA}} \Delta X) \\
& \quad + \sum_{i=\zdim{}+1}^{\xdim{1}+\xdim{2}} \lambda_i(X_{\mathrm{diag}}) - \lambda_i(XX^\top)\\ 
= & \sum_{i=\zdim{}+1}^{\xdim{1}+\xdim{2}} \lambda_i(X_{\mathrm{diag}}) - \lambda_i(XX^\top).
\end{align*}
Finally, we use the matrix perturbation theory \cite{rellich1969perturbation} to calculate the first-order approximation of the effect of $\Delta X$ on the singular values of $X_{\mathrm{diag}}$. The perturbation theory assumes that the perturbation $\Delta X$ is relatively small compared to $X_{\mathrm{diag}}$. Then, we know:
\begin{align*}
\; \| X - D_{\mathrm{DPCA}}E_{\mathrm{DPCA}}X\|_2^2 - \| X - D_{\mathrm{PCA}} E_{\mathrm{PCA}}X\|_2^2 
= & \sum_{i=\zdim{}+1}^{\xdim{1}+\xdim{2}} \lambda_i(X_{\mathrm{diag}}) - \lambda_i(XX^\top) \\
\leq 
& \sum_{i=\zdim{}+1}^{\xdim{1}+\xdim{2}} \lambda_i - \lambda_i - \lambda_ie_i^\top \Delta X e_i \\
= & 
- \sum_{i=\zdim{}+1}^{\xdim{1}+\xdim{2}} \lambda_ie_i^\top \Delta X e_i. 
\end{align*}
\end{proof}

Note that the encoding and decoding matrices of DPCA look like:
$$D_{\mathrm{DPCA}}=
\begin{bmatrix}
D_1 & \mathbf{0} \\
\mathbf{0} & D_1
\end{bmatrix},
E_{\mathrm{DPCA}}=\begin{bmatrix}
E_1 & \mathbf{0} \\
\mathbf{0} & E_2,
\end{bmatrix}$$
where $E_1, E_2, D_1, D_2$ are matrices obtained from each source with DPCA. 

\input{figure_latex/bound}
We examine the correctness of our bound with random data matrices in \figref{fig:bound}. We can see that the gap between DPCA and PCA decreases as the Frobenius norm of $\Delta\X{}$ decreases. The upper bound also has the same trend, while it is always larger than the exact value. Note that in \figref{fig:bound}, all axes are in log scale.


\subsection{Why Robust Task?}
\label{sec:robust_aware}
We now characterize the effect of using task-aware compression and a pre-trained, robust task. We assume that the robust task performs similarly to the original, non-robust task. We also know that the robust task has a lower Lipschitz constant than the non-robust one \cite{RobustNeuralNetworks, CertifiablyRobustNeuralNetworks}.
We denote the robust task model by $\task^*$ and the non-robust task model by $\task$.
We define task-aware autoencoder as 
\begin{align*}
D_{\mathrm{awa}}, E_{\mathrm{awa}} = \arg\min_{D,E} & \quad \|\task^*(x) - \task^*\circ D \circ E(x)\|_2^2 \\
\mathrm{subject\;to} & \quad E(x) \in \mathbb{R}^{\task},
\end{align*}
and task-agnostic autoencoder as
\begin{align*}
D_{\mathrm{agn}}, E_{\mathrm{agn}} = \arg\min_{D,E} & \quad \|x - D \circ E(x)\|_2^2 \\
\mathrm{subject\;to} & \quad E(x) \in \mathbb{R}^{\task},
\end{align*}
where $\circ$ denotes function composition.
For simplicity, we further define 
$$\hat{x}_{\mathrm{awa}} = D_{\mathrm{awa}}\circ E_{\mathrm{awa}}(x) , \quad \hat{x}_{\mathrm{agn}} = D_{\mathrm{agn}}\circ E_{\mathrm{agn}}(x).$$
Then, we prove the following lemma: 
\begin{lemma}[Why task-aware compression and a robust task]
\label{lemma:robust_aware}
Assume robust task model $\task^*$ and non-robust task $\task$ only differ in: 
\begin{equation}
\forall x, ~~ \|\task^*(x)-\task(x)\|\leq\epsilon.
\end{equation}
That is, the robust task and the normal task have a bounded performance gap. Assume that $\task^*$ is a Lipschitz function with constant $L^*$, and $\task$ is a bi-Lipschitz function with constant $L$. Namely,
\begin{equation}
\|\task^*(x)-\task^*(\tilde{x})\|_2 \leq L^*\|x-\tilde{x}\|_2,
\label{lip_task}
\end{equation}
and 
\begin{equation}
\frac{1}{L}\|x-\tilde{x}\|_2 \leq \|\task(x)-\task(\tilde{x})\|_2 \leq L\|x-\tilde{x}\|_2. \label{bilip_task}
\end{equation}
We show that the task losses of task-aware, robust models and task-agnostic, non-robust models are bounded by 
\begin{equation}
\begin{aligned}
& \|\task^*(x) - \task^*(\hat{x}_{\mathrm{awa}})\|_2 - L^*\|x - \hat{x}_{\mathrm{awa}} \|_2 + \frac{1}{L}\|x - \hat{x}_{\mathrm{agn}} \|_2 \\
\leq & \|\task(x)-\task(\hat{x}_{\mathrm{agn}} )\|_2 \\
\leq & \|\task^*(x) - \task^*(\hat{x}_{\mathrm{awa}})\|_2 + 2\epsilon + L^*\|\hat{x}_{\mathrm{awa}} - \hat{x}_{\mathrm{agn}}\|_2 . 
\label{eq:bnd_aware_robust}
\end{aligned}
\end{equation}
\end{lemma}

\begin{proof}
We consider the difference between the two task losses. By the triangle inequality, 
\begin{equation}
\begin{aligned}
& \|\task(x)-\task(\hat{x}_{\mathrm{agn}})\|_2 - \|\task^*(x)-\task^*(\hat{x}_{\mathrm{awa}})\|_2 \\
\leq & \|\task(x) - \task(\hat{x}_{\mathrm{agn}}) - \task^*(x) + \task^*(\hat{x}_{\mathrm{awa}})\|_2 \\
= & \|\task(x) - \task^*(x) + \task^*(\hat{x}_{\mathrm{awa}}) - \task(\hat{x}_{\mathrm{agn}})\|_2 \\
\leq & \|\task(x) - \task^*(x)\|_2 + \|\task^*(\hat{x}_{\mathrm{awa}}) - \task(\hat{x}_{\mathrm{agn}})\|_2 \\
\leq & \epsilon + \|\task^*(\hat{x}_{\mathrm{awa}}) - \task^*(\hat{x}_{\mathrm{agn}}) + \task^*(\hat{x}_{\mathrm{agn}}) - \task(\hat{x}_{\mathrm{agn}})\|_2 \\
\leq & \epsilon + \|\task^*(\hat{x}_{\mathrm{awa}}) - \task^*(\hat{x}_{\mathrm{agn}})\|_2 + \|\task^*(\hat{x}_{\mathrm{agn}}) - \task(\hat{x}_{\mathrm{agn}})\|_2 \\
\leq & 2\epsilon + L^*\|\hat{x}_{\mathrm{awa}} - \hat{x}_{\mathrm{agn}}\|_2 . \label{upper_bnd}
\end{aligned}
\end{equation}
On the other hand, subtracting \Eqref{lip_task} and \Eqref{bilip_task}, we get
\begin{equation}
\begin{aligned}
\|\task^*(x) - \task^*(\hat{x}_{\mathrm{awa}})\|_2 - \|\task(x)-\task(\hat{x}_{\mathrm{agn}} )\|_2 \\
\leq L^*\|x - \hat{x}_{\mathrm{awa}} \|_2 - \frac{1}{L}\|x - \hat{x}_{\mathrm{agn}} \|_2. \label{lower_bnd}
\end{aligned}
\end{equation}
Finally, combining \Eqref{upper_bnd} and \Eqref{lower_bnd}, we get 
\begin{equation*}
\begin{aligned}
& \|\task^*(x) - \task^*(\hat{x}_{\mathrm{awa}})\|_2 - L^*\|x - \hat{x}_{\mathrm{awa}} \|_2 + \frac{1}{L}\|x - \hat{x}_{\mathrm{agn}} \|_2 \\
\leq & \|\task(x)-\task(\hat{x}_{\mathrm{agn}} )\|_2 \\
\leq & \|\task^*(x) - \task^*(\hat{x}_{\mathrm{awa}})\|_2 + 2\epsilon + L^*\|\hat{x}_{\mathrm{awa}} - \hat{x}_{\mathrm{agn}}\|_2 . 
\end{aligned}
\end{equation*}
\end{proof}

\lmaref{lemma:robust_aware} characterizes how close the task losses of task-aware robust models and task-agnostic non-robust models are. The reason that robust task models are preferable to non-robust models is that robust task models have smaller Lipschitz constants. In other words, when noise caused by communication or reconstruction perturbs the input of the models, the output is less sensitive, so the output of the perturbed task is closer to the original output.

With regard to task-aware autoencoders, it is obvious that they are preferable to task-agnostic ones since the former minimizes task losses. Task-agnostic autoencoders aim to reconstruct the full image, but most pixels in an image are not related to the task, so task-agnostic models are more bandwidth inefficient than task-aware models. Of course, when one has sufficient bandwidth to transmit a whole image perfectly, task-agnostic models will perform equally to task-aware models. In this case, $\|x - \hat{x}_{\mathrm{awa}} \|_2 = \|x - \hat{x}_{\mathrm{agn}} \|_2 = 0$ in \Eqref{eq:bnd_aware_robust}. 


\section{The Gaussian CEO Problem}
\label{sec:CEO}
The Gaussian CEO problem \citep{CEO,RateCEO} refers to the problem of distributed inference from noisy observations. The objective is to reconstruct the source from noisy observations rather than the noisy observations themselves, which motivated our first experiment of CIFAR-10 denoising. In the original setting, a White Gaussian source $X$ of variance $P$ is observed through two independent Gaussian broadcast channels $Y_j = X + Z_j$ for $i=1,2$ where $Z_1 \sim \mathcal{N}(0,N_1)$ and $Z_2 \sim \mathcal{N}(0,N_2)$. The observations $Y_1$ and $Y_2$ are separately encoded with the aim of estimating $X$ such that the mean square error distortion between the estimate $\hat{X}$ and $X$ is $D$.

The rate-distortion region $R_{\text{CEO}}(D)$ for the quadratic Gaussian CEO problem is the set of rate pairs $(R_1,R_2)$ that satisfy
\begin{align}
    R_1 &\geq r_1 + \frac{1}{2} \log D - \frac{1}{2} \log \left( \frac{1}{P} + \frac{1 - e^{-2r_2}}{N_2} \right),\\
    R_2 &\geq r_2 + \frac{1}{2} \log D - \frac{1}{2} \log \left( \frac{1}{P} + \frac{1 - e^{-2r_1}}{N_1} \right),
\end{align}
for some $r_1,r_2 \geq 0$ such that 
\begin{align}
    D \geq \left(\frac{1}{P} + \frac{1 - e^{-2r_1}}{N_1} + \frac{1 - e^{-2r_2}}{N_2} \right)^{-1}.
\label{eq:valid_region}
\end{align}

Considering the CIFAR-10 denoising experiment, we have $P=0.3125$, and for a target distortion of PSNR $20$ dB, we have $D=0.01$. For the sake of analysis, we assume the CIFAR-10 source to be Gaussian and find the lower bounds on rates $R_1$ and $R_2$. We begin by solving for the auxiliary variables that satisfy \Eqref{eq:valid_region}. Then, in the region of feasible auxiliary rates, we look for the pair of $(r_1,r_2)$ that minimize the sum lower bound on $R_1+R_2$. Solving this for $N_1=0.01$ and $N_2=1$, we get $R_1 \geq 3.44$ and $R_2 \geq 0.002$. Similarly, for $N_1=0.01$ and $N_2=0.1$, we get $R_1 \geq 2.45$ and $R_2 \geq 0.41$. Under the assumption of Gaussian sources, this clearly demonstrates that the rates for both sources are non-zero. Also, the rate allocated to a source is inversely proportional to the noise. Therefore, $R_1 > R_2$ when source $1$ is less noisy, implying that higher bandwidth is allocated to source $1$ since it contains more \textit{information} and is more \textit{important}.

\section{NDPCA with $4$ sources: }
\label{sec:4_sources}
To showcase the capability of NDPCA under more than $2$ sources, we examine it on the most complicated dataset amomng the $3$–Airbus detection. 
Views $1$ and $2$ have resolutions of $(160 \times 224)$ pixels, whereas views $3$ and $4$ have $(288 \times 224)$ pixels. Same as the previous experiments, we intentionally set the views to different sizes so that the importance to the task is unequal, resulting in different bandwidth allocatation among the sources in Fig. \ref{fig:4_src}.
\input{figure_latex/4_sources}







%% file: figure_latex/bound.tex
\begin{figure*}[ht]
\centering
\includegraphics[width=0.5\textwidth]{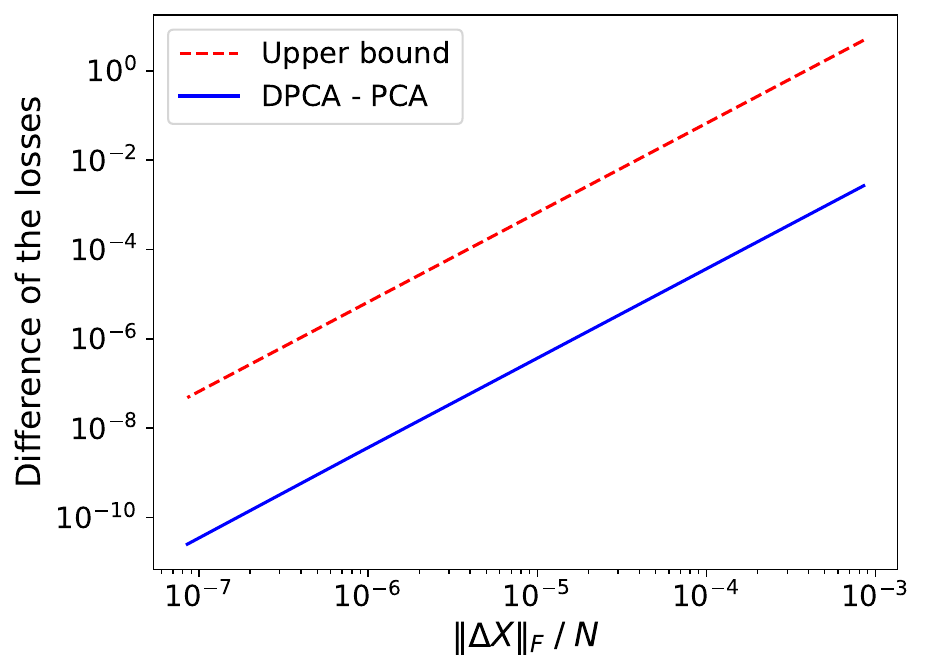}
\caption{
\small{\textbf{Bound from \lmaref{lemma:boundDPCA}:} The obtained upper bound is always larger than the difference of losses of DPCA and PCA.}}
\label{fig:bound}
\vspace{-1em}
\end{figure*}

%% file: figure_latex/4_sources.tex
\begin{figure*}[h]
\centering
\includegraphics[width=0.8\textwidth]{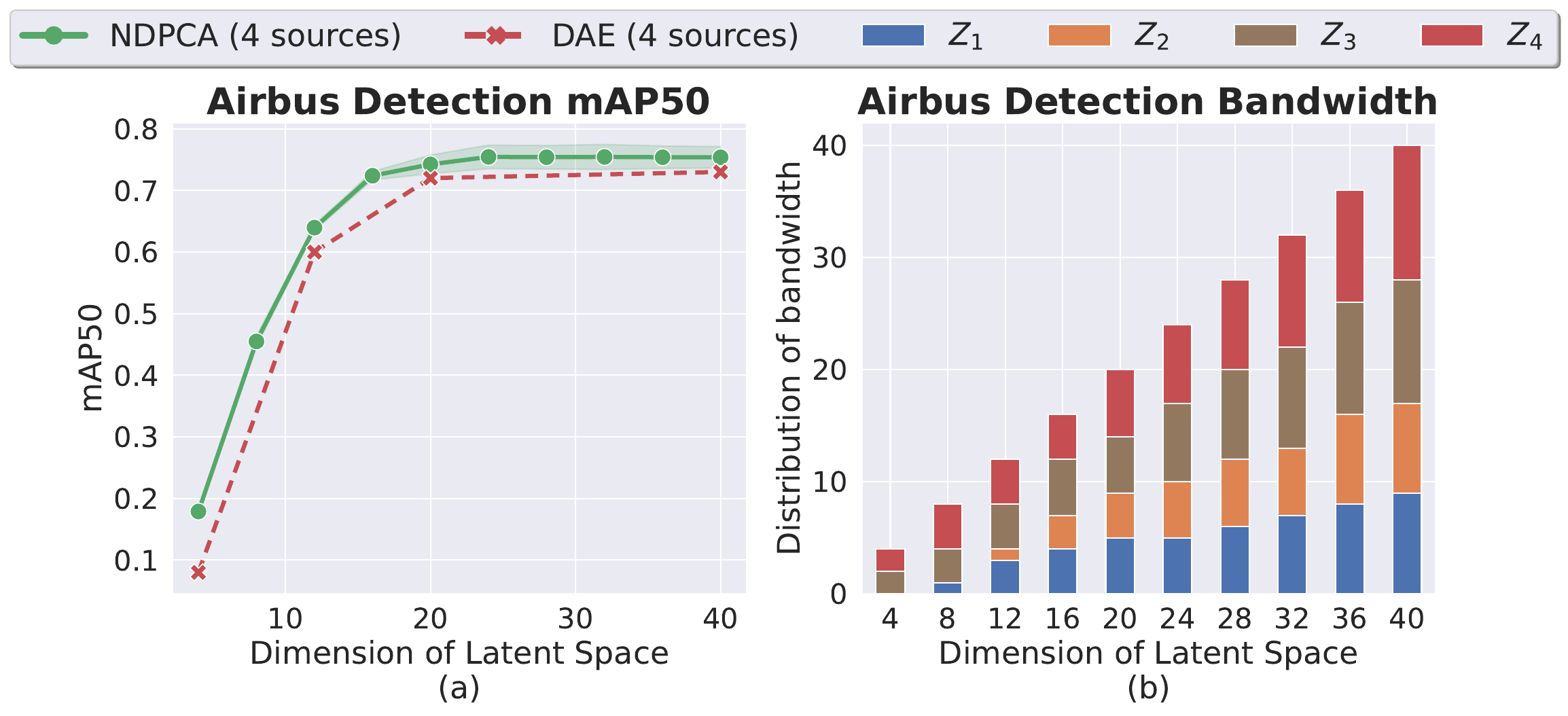}
\caption{\small{
\textbf{NDPCA with $4$ data sources:} (a) The performance of NPCA compared to DAE with $4$ data sources. NDPCA is better than DAE, which aligns with our result with $2$ data sources. 
(b) Distribution of total available bandwidth (latent space) among the $4$ views for NDPCA. 
The difference in resolution emphasizes the distinct importance of each view in object detection, therefore $Z_3$ and $Z_4$ have greater dimensions than $Z_1$ and $Z_2$. 
}}
\label{fig:4_src}
\vspace{-1.5em}
\end{figure*}

%% file: section/details.tex
\section{Details of the Datasets}
\label{sec:app_datasets}

\subsection{CIFAR-10 denoising:}
We started with the standard CIFAR-10 dataset and normalized the images to $[0,1]$. Two different views are created by adding different levels of Gaussian noise, $\mathcal{N}(0,0.1^2)$ and $\mathcal{N}(0,1)$. The pre-trained task model is created by training a denoising autoencoder that takes both views, concatenates them along the channel dimension, and produces a clean image. The autoencoders need to learn features that are important for this task model.
\label{sec:app_datasets_pnp}

\subsection{Locate and lift:}
We collected $20,000$ pairs of actions and the corresponding images of both views for our training set. 
The actions are $4$ dimensional, controlling the $x,y,z$ coordinate movements and the gripper of the robotic arm. 
We randomly cropped the images from $128 \times 128$ to $112 \times 112$ pixels to make our autoencoder more robust. 
The expert agent is pre-trained by the same data augmentation as well. 

\subsection{Airbus detection:}
We first cropped all original images of $2560 \times 2560$ pixels (Fig. \ref{fig:airbus_orig}) into $224 \times 224$ pixels with $28$ pixels overlapping between each cropped image. 
We then eliminated the bounding boxes that are less than $30\%$ left after cropping. 
\input{figure_latex/airbus_2560image}

\section{Implementation Details}
\label{sec:app_imple}

\subsection{CIFAR-10 denoising:}
For the CIFAR-10 dataset, we used the standard CIFAR-10 dataset and applied different levels of AWGN noise to create two correlated datasets. We used the CIFAR-10 experiments as a proof of concept to try different architectures and loss functions and other techniques to finalize our framework. We choose $\lambda_{\mathrm{task}}=1$ for the task-aware setting and $\lambda_{\mathrm{rec}}=1$ for the task-agnostic setting. We run $4$ random seeds on NDPCA and all baselines to evaluate the performance. 

\subsection{Locate and lift:}
For the locate and lift experiment, we trained our autoencoder with the same random cropping setting as in Sec. \ref{sec:app_datasets}, which cropped the images from $128 \times 128$ to $112 \times 112$ pixels. During testing, we randomly initialized the location of the brick and center-cropped the images from $128 \times 128$ to $112 \times 112$ pixels. 
We scaled all images to $0$ to $1$ and ran $5$ random seeds on NDPCA and all baselines to evaluate the performance. 
For the task-aware setting, $\lambda_{\mathrm{task}}=500$, and $\lambda_{\mathrm{rec}}=1000$ for the task-agnostic. setting 

\subsection{Airbus detection:}
For the Airbus detection task, we used the original Yolo paper for our object detection model together with the detection loss \cite{Yolo}. Our experiments with the latest state-of-the-art Yolo v8 model \cite{Yolov8} showed that there is no big difference in the Airbus detection dataset in terms of run time and accuracy. 
Since the size of the original dataset is not enough to train an object detection model, we used the data augmentation proposed in Yolo v8, mosaic, to increase the size of the dataset. Mosaic randomly crops $4$ images and merges them to generate a new image. 
We used random resized crop, blur, median blur, and CLAHE enhancement during training, each with probability 0.05 by functions in the Albumentations package \cite{Albumentations}. 
We increased the size of the Airbus dataset from $5904$ to $21808$ with mosiac and trained the Yolo detection model. Finally, we trained our autoencoder with the same dataset, but downsample the images to $112 \times 112$ pixels so that the autoencoder is faster to train. 
For the task-aware setting, $\lambda_{\mathrm{task}}=0.1$, and $\lambda_{\mathrm{rec}}=0.5$ for the task-agnostic setting. 
We run $2$ random seeds on NDPCA and all baselines to evaluate the performance.

\subsection{Neural Autoencoder Architecture and Hyperparameters}
\input{figure_latex/ResAutoncoder}
We used the ResNet encoder shown in \figref{fig:resnet_enc} and the decoder in \figref{fig:resnet_dec} for all experiments. We used different numbers of filters and numbers of residual blocks for our experiments, shown as $C$ and $r$. We denote $\zdim{}$ as the number of latent dimensions.
The numbers of filters are $C_1=32,~ C_2=64,~ C_3=128$, $C_1=8,~ C_2=16,~ C_3=32,~ C_4=64$, and $C_1=16,~ C_2=32,~ C_3=64,~ C_4=128$, and the numbers of residual blocks are $r=0$, $r=1$, $r=1$  for CIFAR-10 denoising, locate and lift, and Airbus detection. For CIFAR-10 denoising, we use the Adam optimizer with a learning rate of $0.0002$, and for the other two experiments, we use the Adam optimizer with a learning rate of $0.0001$.
For the sake of training speed, when training DAE and JAE, we first trained a large network with $\zdim{\mathrm{max}}$ with each random seed. Then, we fixed the network parameters and trained concatenate $3$ fully connected layers on each encoder and decoder network to compress and decompress the data to smaller $\zdim{}$.

\subsection{Balancing Task-aware and Task-agnostic Loss}
\label{sec:weighted_task_loss}
\input{figure_latex/weighted_task-aware}
NPDCA has a loss function consisting of $2$ terms, as shown in \eqref{eq:overall_loss}:
\begin{equation*}
\mathcal{L}_{\mathrm{tot}} = 
\lambda_{\mathrm{task}} \underbrace{\|\Yhat - \Y \|_F^2}_{\mathrm{task~loss}}
+ \lambda_{\mathrm{rec}} \underbrace{\left(\|\Xhat{1} - \X{1}\|_F^2 + \|\Xhat{2} - \X{2}\|_F^2 + \dots \|\Xhat{K} - \X{K}\|_F^2 \right)}_{\mathrm{reconstruction~loss}}.
\tag{\ref{eq:overall_loss} revisited}
\end{equation*}
Previously, we tested two extreme cases of \eqref{eq:overall_loss}: task-aware when $\lambda_{\mathrm{task}}>0,\lambda_{\mathrm{rec}}=0$, and task-agnostic when $\lambda_{\mathrm{task}}=0,\lambda_{\mathrm{rec}}>0$. Of course, one can use different weighted sums of the $2$ terms in \eqref{eq:overall_loss}, which we call weighted task-aware. We show the resulting reconstructed image in \figref{fig:weighted_task_loss}, whose weights are a mixture of half of the two other methods.
Weighted task-aware images have both blurry reconstructions of the original images and task-relevant features. Unsurprisingly, the task loss and the reconstructed loss of weighted task-aware images are between pure task-aware and task-agnostic, that is, we can use the weights in the loss function to trade off compressing human perception features against task-relevant features. 
Interestingly, we can see that the task-aware images look similar to the images without Airbuses (last $2$ columns), and when there are Airbuses, the task-aware images look different. 
It means that the features of no Airbuses are pretty much the same in the latent space, thus resulting in similar images in pixel space.
Hence we can conclude that task-aware features are not random noise, they are meaningful features only to the task model but not to our eyes. 

\subsection{Storage and Training Complexity}
\input{figure_latex/storage_table}
One key feature of NDPCA is that it only needs one model to operate in different bandwidths. Therefore, we only need to train and store one model at the edge devices and the central node. We compare the complexity of storage and training in \tabref{tab:storage}. 
Although NDPCA has a larger storage size and longer training time than other models, it can operate across different bandwidths. 
According to \tabref{tab:storage}, if all models operate in more than $1$ bandwidths, NDPCA saves more storage and training overload because other models have more than $50\%$ of NPDCA's overload.
For CIFAR-10 denoising, we tested the training time on an RTX 4090, and for the locate and lift and Airbus detection experiments, we tested the training time on an NVIDIA RTX A5000.

\section{Ablation Study}
\label{sec:ablation}
\input{figure_latex/norm_results}
\subsection{Cosine similarity and nuclear norm}
\label{sec:ablation_norm}
In \figref{fig:norm_results}, we show that adding nuclear norm or cosine similarity in the training loss \Eqref{eq:overall_loss} does not help the model perform when we use DPCA to project latent representations into lower dimensions. 
We compared our proposed NDPCA with the DPCA module against NDPCA without the DPCA module but with the penalization of the nuclear norm and cosine similarity added. The weights of all the additional terms are $0.1$. 
From \figref{fig:norm_results}, we conclude that the DPCA module can increase the performance better than the other two.

\input{figure_latex/ablation_results}
\subsection{DPCA module}
\label{sec:ablation_DPCA}
In \figref{fig:ablation_results}, we show that the proposed DPCA module can help the neural autoencoder learn linear compressible representations, as described in \secref{sec:NDPCA}. We see that with the DPCA module, NDPCA can increase the performance in lower bandwidths, while saturating at the performance close to the model without the module. 
We conclude that with the DPCA module, NDPCA learns to generate low-rank representations, so the performance is better in lower bandwidths. However, when the bandwidth is higher, the bandwidth can almost fully restore the representations, so the two methods perform similarly.

\subsection{Single view performance of locate and lift}
\label{sec:ablation_lift}
In the locate and lift experiments, the reinforcement learning agent leverages information from both views as input to manipulate. Here, we detail why the $2$ views are complementary to accomplish the task. The success rate of an agent is $76\%$ with only the arm-view and $45\%$ with the side-view. When combining both, the success rate is $83\%$. 
The reason why the views are complementary is that the side-view provides global information on the position of the arm and the brick, but sometimes the brick is hidden behind the arm. The arm-view captures detailed information from a narrow view of the desk. Once the arm-view captures the brick, it is straightforward to move toward it and lift it. 
The arm view is more important because with only the arm-view, the agent can randomly explore the brick, but with only the side-view, the brick might be vague to see and thus harder to lift. 
Of course, with both views, the robotic arm can easily move toward the vague position of the brick and use arm-view to lift it.

%% file: figure_latex/airbus_2560image.tex
\begin{figure*}[ht]
\centering
\includegraphics[width=0.75\textwidth]{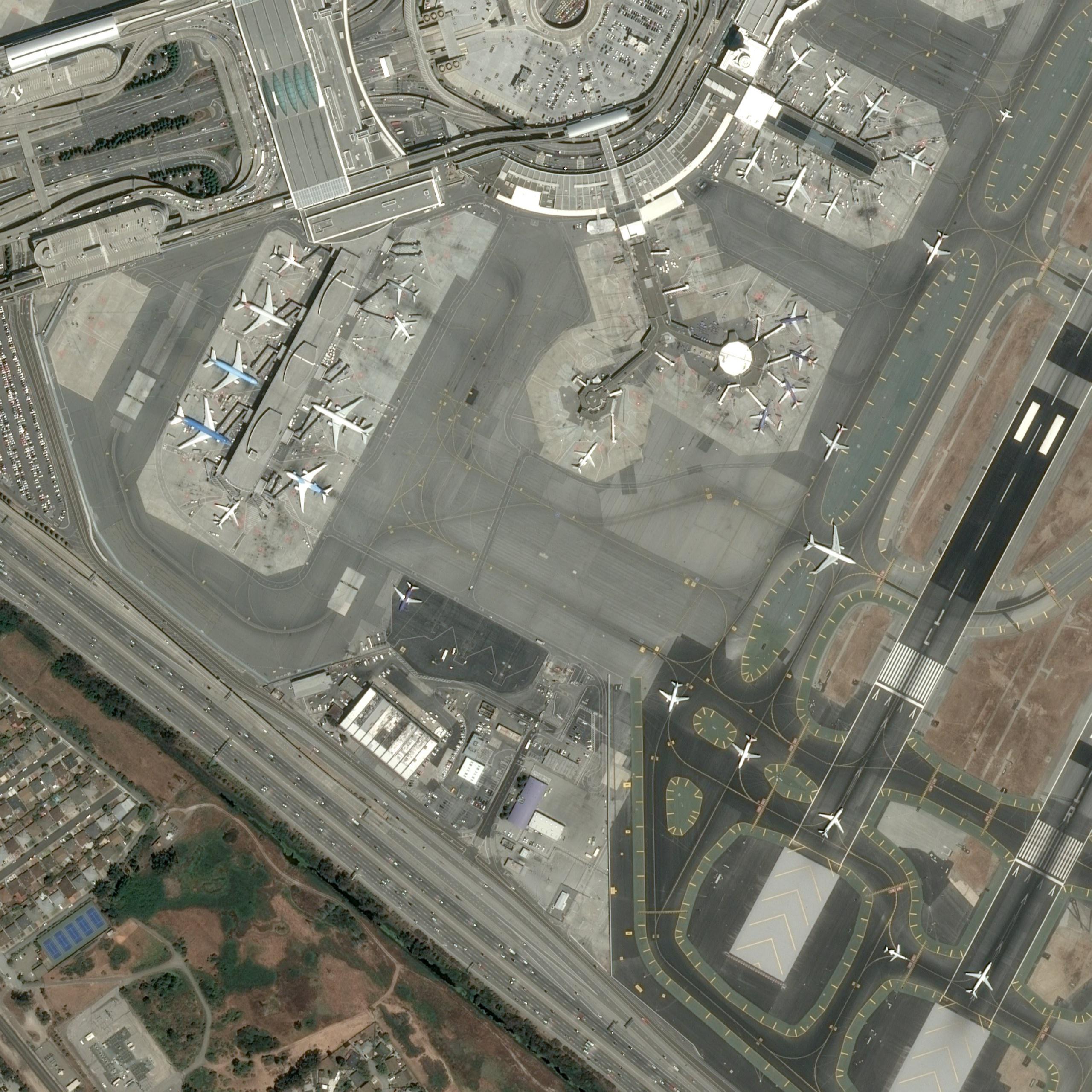}
\caption{
\small{\textbf{Original image of airbus detection}. The original images are $2560\times2560$ pixels, and we cropped them into smaller pieces in $224\times224$.}}
\label{fig:airbus_orig}
\vspace{-1em}
\end{figure*}

%% file: figure_latex/ResAutoncoder.tex
\begin{figure*}[ht]
\centering
\begin{subfigure}{0.4\textwidth}
\includegraphics[width=\textwidth]{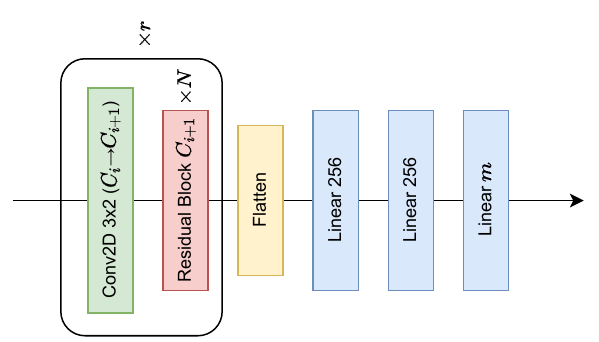}
\caption{
\small{Encoder architecture.}}
\label{fig:resnet_enc}
\end{subfigure}
\begin{subfigure}{0.51\textwidth}
\centering
\includegraphics[width=\textwidth]{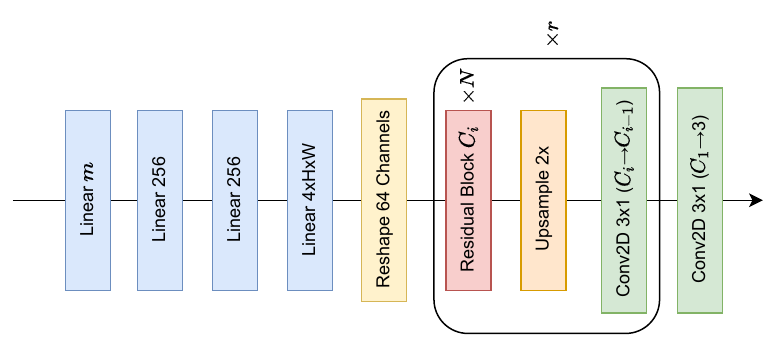}
\caption{
\small{Decoder architecture.}}
\label{fig:resnet_dec}
\end{subfigure}
\caption{\small{\textbf{ResNet Autoencoer:} The encoder processes inputs through $r$ convolution layers and $r\times N$ residual blocks, followed by $3$ fully connected layers with ReLU activation. 
The decoder processes latent representations in the reverse order from the encoder with $2\times$ upsamplings.  
}}
\label{fig:resnet}
\end{figure*}

%% file: figure_latex/weighted_task-aware.tex

\begin{figure*}[t!]
\centering
\includegraphics[width=0.85\textwidth]{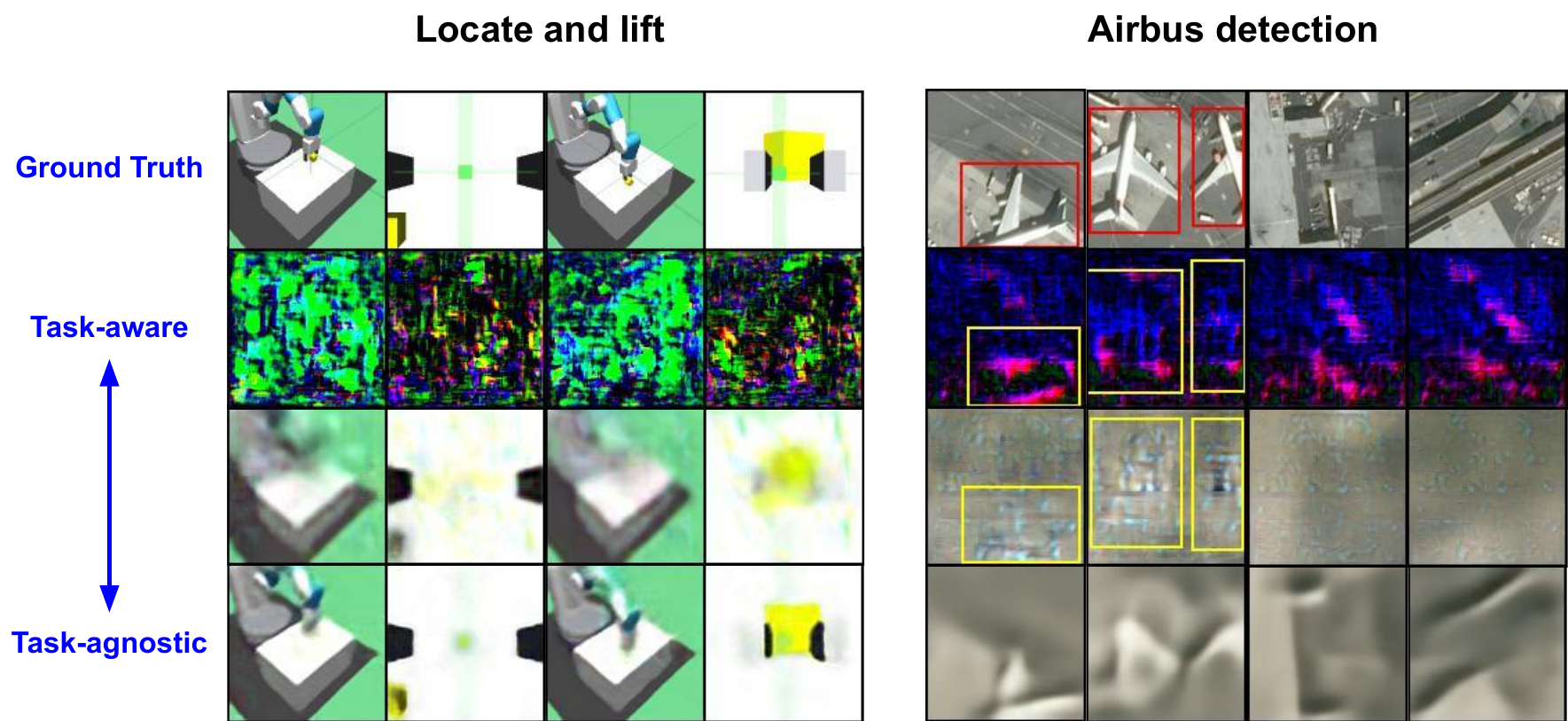}
\caption{\small{\textbf{Weighted task-loss:} 
    Weighted task-aware images faintly reconstruct the original images while restoring task-relevant features with high-frequency noise. In Airbus detection, location of Airbuses is captured with shiny high-frequency pixels in row $3$.
    }}
\label{fig:weighted_task_loss}
\end{figure*}

%% file: figure_latex/storage_table.tex

\begin{table}[!htbp]
\small
\centering
\begin{booktabs}{
  colspec = {crrrrrr},
  cell{1}{2,4,6} = {c=2}{c}, 
}
\toprule
  Model & CIFAR-10 & & Locate and lift & & Airbus detection &  \\
\midrule
       &  Storage (MB) &  Train (hr) &  Storage (MB) &  Train (hr) &  Storage (MB) &  Train (hr) \\
\cmidrule[lr]{2-3}\cmidrule[lr]{4-5}\cmidrule[lr]{6-7}
    NDPCA & $8.3$ & $0.25$ &$16.4$ & $5.0$ & $33.0$ & $13.0$ \\
    DAE  & $5 \times 8.4$ & $5 \times 0.21$ & $4\times16.3$ & $4\times5.0$ & $4\times22.5$ & $4\times11.5$ \\
    JAE  & $5 \times 10.2$ & $5 \times 0.22$ & $4\times11.4$ & $4\times3.5$ & $4\times32.9$ & $4\times10.5$ \\
\bottomrule
\end{booktabs}
\vspace{0.25em}
\caption{\small{\textbf{Storage and training complexity:} NDPCA has slightly more storage and training overload than other models for a single bandwidth but can operate across different bandwidths. We multiply the number of bandwidths tested in \figref{fig:results} to the storage size and training time of DAE and JAE as they require different models for different compression levels. 
}}
\label{tab:storage}
\end{table}

%% file: figure_latex/norm_results.tex
\begin{figure*}[ht]
\centering
\includegraphics[width=1\textwidth]{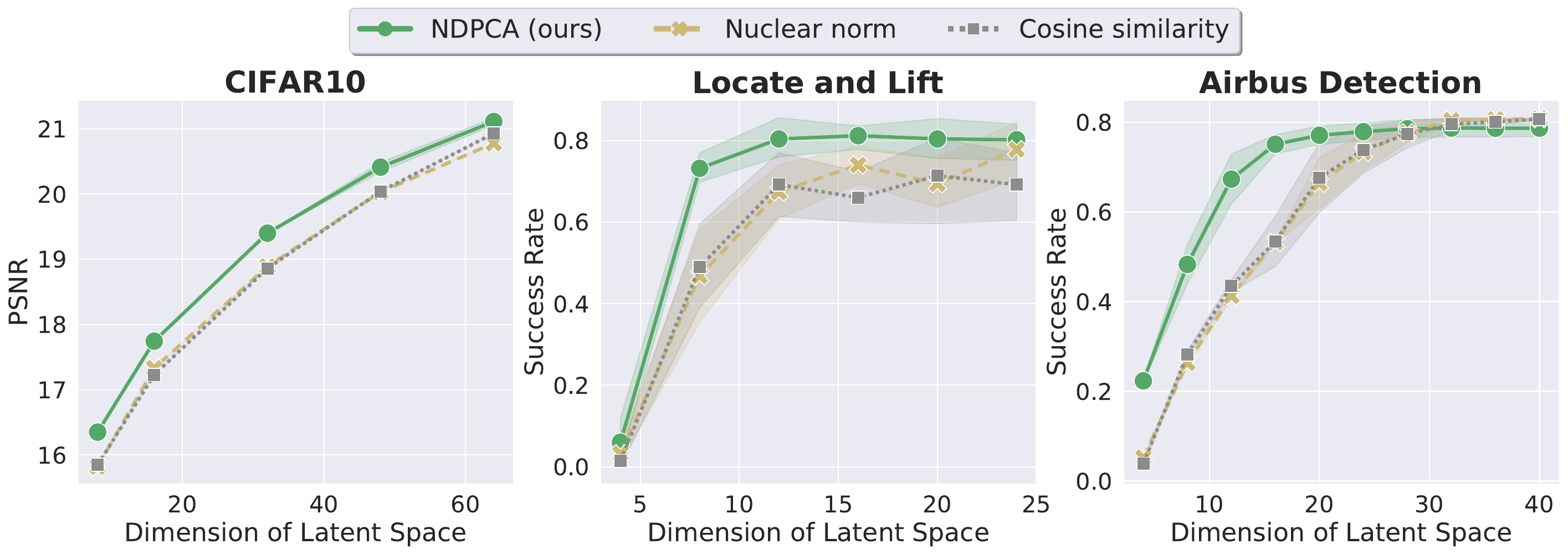}
\caption{\small{\textbf{Ablation study of the nuclear norm and cosine similarity:} Adding the nuclear norm or cosine similarity to the loss function does not improve the performance of the model when compressing latent representations to lower dimensions.}}
\label{fig:norm_results}
\end{figure*}

%% file: figure_latex/ablation_results.tex
\begin{figure*}[ht]
\centering
\includegraphics[width=1\textwidth]{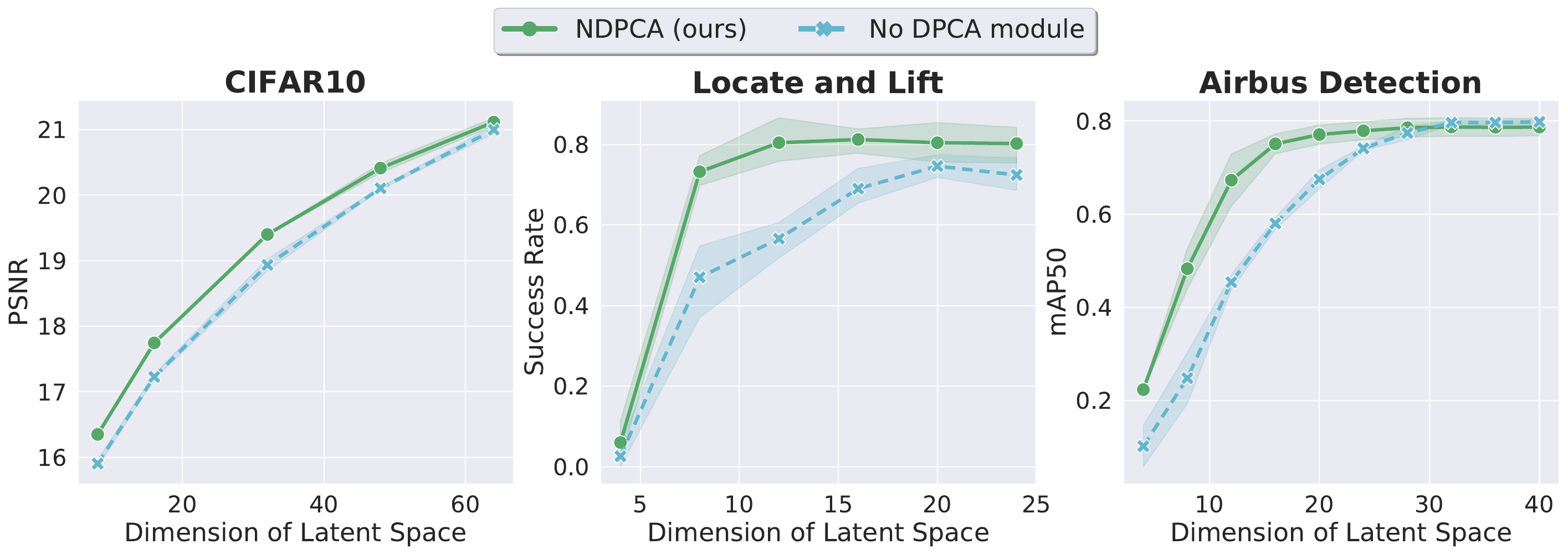}
\caption{\small{\textbf{Ablation study of DPCA module:} The proposed DPCA module effectively increases the performance in lower bandwidths, while achieving the same performance at larger bandwidths.}}
\label{fig:ablation_results}
\end{figure*}